\documentclass[journal]{vgtc}                     

\onlineid{1780}

\vgtccategory{Research}

\vgtcpapertype{algorithm/technique}

\title{Simplicial Approximation of Deforming 3D Spaces\\for Visualizing Fusion Plasma Simulation Data}

\author{%
  \authororcid{Congrong~Ren}{0009-0006-6285-7271}
  and 
  \authororcid{Hanqi~Guo}{0000-0001-7776-1834}
}

\authorfooter{
  \item Congrong~Ren and Hanqi~Guo are with The Ohio State University. E-mail: \{ren.452\,$|$\,guo.2154\}@osu.edu\,.
}

\abstract{%
  We introduce a fast and invertible approximation for data simulated as 2D planar meshes with connectivities along the poloidal dimension in deforming 3D toroidal (donut-like) spaces generated by fusion simulations. In fusion simulations, scientific variables (e.g., density and temperature) are interpolated following a complex magnetic-field-line-following scheme in the toroidal space represented by a cylindrical coordinate system. This deformation in 3D space poses challenges for visualization tasks such as volume rendering and isosurfacing. To address these challenges, we propose a novel paradigm for visualizing and analyzing such data based on a newly developed algorithm for constructing a 3D simplicial mesh within the deforming 3D space. Our algorithm introduces no new nodes and operates with reduced time complexity, generating a mesh that connects the 2D meshes using tetrahedra while adhering to the constraints on node connectivities imposed by the magnetic field-line scheme. In the algorithm, we first divide the space into smaller partitions to reduce complexity based on the input geometries and constraints on connectivities. Then we independently search for a feasible tetrahedralization of each partition taking nonconvexity into consideration. We demonstrate use cases of our method for visualizing XGC simulation data on ITER and W7-X.
}


\teaser{
  \centering
  \includegraphics[width=0.98\linewidth]{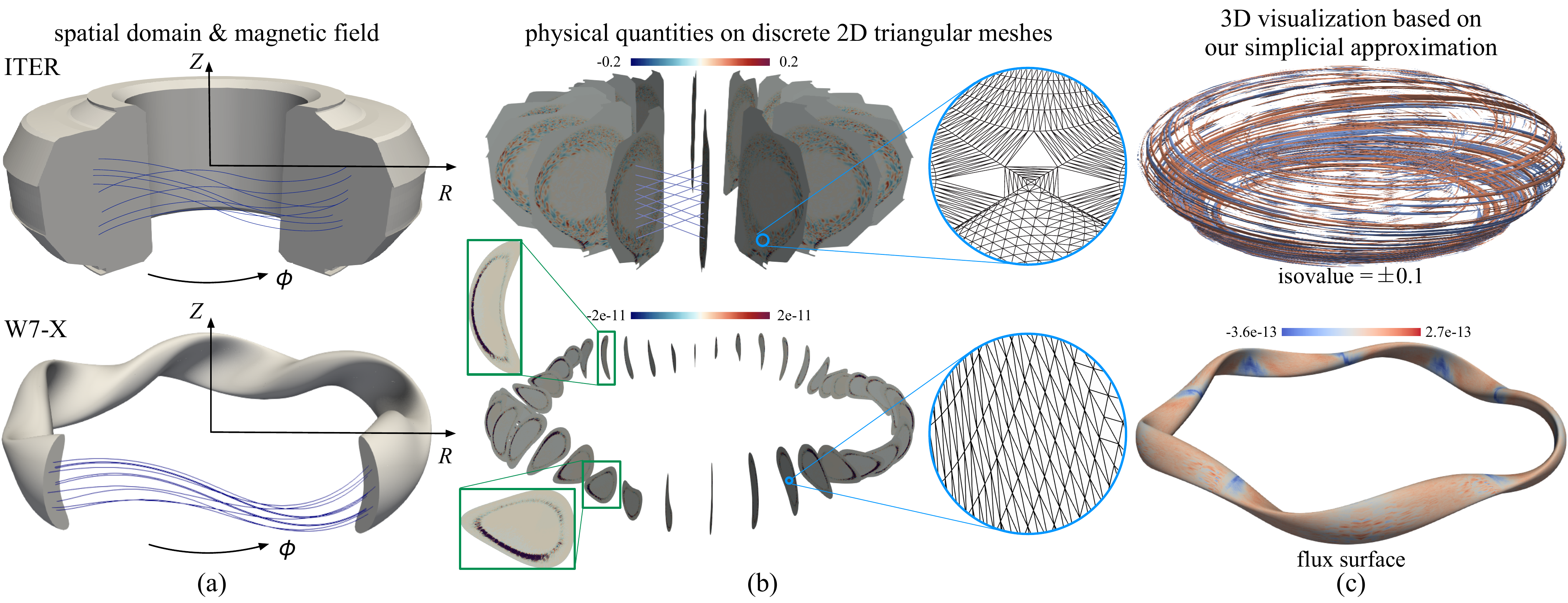}
  \caption{
   Illustrations and isosurfacing results of XGC fusion plasma simulation data for tokamaks (top) and stellarators (bottom) in the cylindrical coordinate system $(R,Z,\phi)^\intercal$, where $R$, $Z$, and $\phi$ represent radial, axial, and toroidal coordinates, respectively. (a) Simulation domain and magnetic field lines in two devices, the International Thermonuclear Experimental Reactor (ITER)~\cite{aymar2002iter,mukhovatov2003overview} and Wendelstein 7-X (W7-X)~\cite{erckmann2007electron,geiger2014physics}, based on the concepts of tokamaks and stellarators, respectively. ITER exhibits a uniform shape along the toroidal axis, while W7-X features a twisting geometry. (b) XGC output data, consisting of (1) 2D triangular meshes on discrete poloidal planes and (2) connectivities between mesh nodes in every pair of adjacent planes approximating the magnetic lines (only shown in the ITER case). (c) Isosurfacing results derived from XGC data by our tetrahedralization algorithm. The upper shows isosurfaces of a scalar function, \textit{normalized electrostatic potential perturbation} ($\delta n_e/n_{e0}$) in ITER. The lower is a isosurface of \textit{poloidal magnetic flux} ($\Psi$) (a.k.a., \textit{flux surface}) in W7-X, colored by \textit{perturbed electrostatic potential} ($\delta\Phi$).
  }
  \label{fig:teaser}
}

\graphicspath{{figs/}{figures/}{pictures/}{images/}{./}} 

\usepackage{tabu}                      
\usepackage{booktabs}                  
\usepackage{lipsum}                    
\usepackage{mwe}                       
\usepackage{mathptmx}                  
\usepackage{amsmath}
\usepackage{amssymb}
\usepackage{amsthm}
\usepackage{multirow}
\usepackage{multicol}
\usepackage{algorithm2e}
\usepackage{float}
\usepackage{enumitem}
\usepackage{thmtools,thm-restate}
\usepackage{wrapfig}
\usepackage{makecell}

\SetKwComment{Comment}{/* }{ */}
\RestyleAlgo{ruled}

\begin{document}


\firstsection{Introduction}
\label{sec:intro}
\maketitle
Fusion energy research, particularly in plasma physics, seeks to achieve sustained nuclear fusion by controlling high-temperature plasma, demanding a deep understanding of the complex dynamics involved~\cite{han2022sustained,linke2019challenges,miyamoto2005plasma}. The X-Point Included Gyrokinetic Code (XGC) simulates plasma behavior in fusion devices such as tokamaks and stellarators, which use magnetic fields to confine plasma within a torus geometry~\cite{chang2006integrated,d2017fusion}. In tokamaks, the magnetic field follows a uniform confinement around the torus (\Cref{fig:teaser} (a) upper), while stellarators use more intricate, twisted magnetic fields to stabilize the plasma (\Cref{fig:teaser} (a) lower)~\cite{xu2016general}. In both configurations, plasma moves along helical magnetic field lines, leading to complex deformation patterns of physical quantities like temperature and density~\cite{helander2012stellarator}. In XGC simulation data, science variables are modeled on 2D triangular meshes along the toroidal dimension (\Cref{fig:teaser} (b))~\cite{xgc_mesh}. As a result, both the symmetrical field of the tokamak and the twisted geometry of the stellarator form deforming 3D spaces.

In our observation, to date, deforming space is either overlooked or induces high complexity in scientific visualization. In the former case, physical insights between discrete samples are ignored because usually the deforming data are visualized only at these discrete planes. In the latter case, deforming space challenges visualization and analysis tasks. For example, magnetic-line-following interpolation incurs prohibitively high computational costs in visualization tasks such as volume rendering: interpolating science variables (e.g., density and temperature) in fusion plasma requires tracing streamlines in a magnetic field~\cite{white2017theory}. Furthermore, such a nonlinear interpolation basis makes it even more difficult for root-finding tasks, such as extracting isosurfaces where the interpolated value equals a given value and finding critical points where interpolated vector field values vanish. Later sections will refer to the interpolation and root-finding processes as \emph{forward} and \emph{inverse approximations}, respectively, as we further motivate visualization with deforming space.

In this work we explore a novel paradigm to represent deforming space by (simplicial) 3D meshes to simplify the visualization and analysis of fusion simulation data. Our simplicial mesh offers a fast, continuous, and invertible data representation of spatially varying data, making forward and inverse approximations possible with less computational cost and facilitating visualization tasks, including volume rendering and isosurfacing, using off-the-shelf visualization tools such as ParaView~\cite{paraview}. Although our primary application driver is fusion plasma, our methodology is applicable also in other scientific domains, such as solid mechanics, where researchers encounter challenges in modeling the deformation of elastic materials undergoing torsion or bending (e.g., metals or rubber)~\cite{rayhan2020modeling} and molecular dynamics, where protein folding and DNA twisting involve continuous structural deformations~\cite{gromiha1996anisotropic,yang1993finite}.

Formally, the inputs are a series of 2D simplicial meshes along toroidal axis, each representing the spatial discretization at some toroidal angle $\phi$. The assumption reflects real application scenarios: most simulations output data as a simplicial mesh. The deformation of space is explicitly defined as connectivities of mesh nodes in adjacent steps (i.e., the ``next node'' in fusion applications, as explained in~\Cref{sec:preliminaries}). In other words, these node connectivities are a discrete representation of how spatial locations are transformed in adjacent steps. The connectivities may be available from data~\cite{xgc_mesh} or derived by tracing particle trajectories based on the underlying physics. 3D deforming meshing from 2D triangular meshes and node connectivities then imposes multiple geometrical constraints on the resulting mesh: (1) the mesh includes all simplices (e.g., triangles, edges, and nodes) from the input spatial meshes without adding new nodes (i.e., \textit{Steiner points}); (2) the mesh contains all edges that connect nodes in adjacent steps, as specified by node connectivities; (3) the mesh is simplicial (containing only tetrahedral cells) for visualization and analysis tasks.

A new methodology is required to search for a triangulation of a deforming space subject to the above-described constraints with a reasonable time cost, considering the complexity introduced by these constraints. General-purpose triangulation tools like TetGen~\cite{hang2015tetgen} and Gmsh~\cite{geuzaine2009gmsh}, which rely on Delaunay triangulation~\cite{chew1987constrained}, pose two key challenges. First, these tools have significant time complexity, typically scaling as $O(|\mathcal{V}|\log|\mathcal{V}|)$~\cite{chew1987constrained}, where $|\mathcal{V}|$ represents the number of input vertices. In the case of deforming 3D meshes, this complexity can escalate to $O(|\mathcal{V}|^2)$ due to degenerate cases where all nodes lie on two planar meshes, complicating the determination of a unique Delaunay triangulation~\cite{su1995comparison}. Second, these tools introduce unnecessary Steiner points to satisfy Delaunay conditions, resulting in computational overhead for evaluating values at these points. However, Delaunay conditions are unnecessary, as the connectivity constraints already ensure the accuracy of the approximation.

In this paper, we propose a divide-and-conquer algorithm that divides the space into small and independent partitions (the \textit{divide} stage) and then triangulates each partition (the \textit{conquer} stage). The divide stage aims to reduce computational costs by decomposing the problem into independent and manageable subproblems of triangulating partitions.  For shared boundaries of neighboring partitions, we define a rule to triangulate the boundaries so that each partition can be processed independently. The conquer stage focuses on each individual partition, which is triangulated by iterating over nodes and eliminating all tetrahedra that connect to a node. A decision-tree-based search is applied to decide which node to eliminate and how to triangulate its surrounding with all geometry (e.g., nonconvexity) and connectivity constraints considered.  Once we triangulate all partitions, the resulting mesh can enable and simplify many visualization tasks, and we demonstrate the use cases of our meshes with fusion plasma data. In summary, our framework makes the following contributions:
\begin{itemize}[topsep=1ex,itemsep=-1ex,partopsep=1ex,parsep=1ex]
    \item A novel paradigm to simplify visualization of scientific data from fusion plasma simulations;
    \item A practical meshing algorithm for deforming 3D space with constraints including a series of 2D simplicial meshes and node connectivities between adjacent meshes.
\end{itemize}

\section{Background and Related Work}
\label{sec:related}

We review several basic concepts in triangulation and the most relevant literature on it. While few directly tackle triangulation in deforming space, triangulation (primarily in 2D and 3D spaces for real-world applications) is widely studied in geometric modeling, computer graphics, and computational sciences; readers are referred to the work of De Loera et al.~\cite{de2010triangulations} for a comprehensive understanding.

\textbf{Simplices and simplicial prisms}. An \textit{$n$-simplex} is a convex hull with $n+1$ nodes that are affinely independent in $\mathbb{R}^n$. For example, a 0-simplex is a point, a 1-simplex is a line segment, a 2-simplex is a triangle, and a 3-simplex is a tetrahedron. An \textit{$(n+1)$-D simplicial prism} (also referred to as $(n+1)$-prism or prism) is derived by extruding an $n$-simplex to one dimension higher (\Cref{fig:deforming}).  Note that prisms are usually nonsimplicial; for example, a 3D triangular prism consists of two congruent triangles and three quadrilaterals. 

\textbf{Polytopes and their triangulation}. An \textit{$n$-polytope} is a geometry object in $\mathbb{R}^n$ whose faces are all \textit{flats}~\cite{sommerville2020introduction}. A triangulation of an $n$-polytope is a subdivision of the polytope into a finite collection of $n$-simplices such that the union of all simplices is equal to the polytope (\textit{union property}) and intersection of any pair of simplices is a common $(n-1)$-facet or empty (\textit{intersection property})~\cite{de2010triangulations}. We focus on 3-polytopes, namely,  \textit{polyhedra}, that may or may not be convex in this study. Triangulation of a polyhedron is also referred to as \textit{tetrahedralization} in the following.

\textbf{Triangulation problems related to this research}. We relate our research to three notoriously challenging triangulation problems: (1) triangulation without additional (Steiner) points~\cite{si2019simple}, (2) triangulation with low time complexity, and (3) triangulation with internal edge constraints~\cite{bajaj1996surface,bajaj1999tetrahedral,hang2015tetgen}. First, we aim to avoid Steiner points, as data beyond the nodes of the spatial mesh are often not readily available. Triangulation without Steiner points introduces significant complexity and is not always achievable; for instance, there exist nontriangulable polytopes, such as the \textit{Sch\"{o}nhardt polyhedron} illustrated in~\Cref{fig:staircase}~{(a)}~\cite{de2010triangulations}. This polytope can be constructed by rotating one of two parallel equilateral faces of a 3D triangular prism and inserting opposite diagonals in previous rectangles. Second, the time complexity associated with triangulating a nonconvex polytope is considerable, with estimates approaching $O(|\mathcal{V}|^2)$~\cite{geuzaine2009gmsh,hang2015tetgen} and potentially reaching $O(|\mathcal{V}|^3)$~\cite{chazelle1990triangulating,si2019simple}. Third, existing research on 3D triangulation predominantly emphasizes the constraints imposed by the surface triangular mesh of a bounded 3D domain, often neglecting interior edge constraints within the domain. 

\textbf{Triangulation in computational sciences}. Meshing 4D spacetime (3D space plus 1D time) in computational science, which has been recently explored to simulate partial differential equations directly without using traditional 3D meshes and timestepping methods~\cite{erickson2005building, ishii2019solving}, is related to but has fundamentally different goal from this paper; we attempt to represent science variables available on spatial meshes, which are still the prevailing way to represent field data in today's computational sciences. While native spacetime meshes may interpolate and represent time-varying variables as conventional meshes, spacetime simulations are subject to high complexity, increased memory footprint, and numerical instabilities. For example, in finite element methods (FEMs), one can create and update an $(n+1)$-D spacetime mesh by adding $(n+1)$-simplices on an $n$-D spatial mesh along the time dimension~\cite{thite2009adaptive,wang2015high,tavelli2016staggered}. This method forms a simplicial complex with a terrain-like surface as spacetime. Researchers have expressed interest in different phenomena, such as wave propagation~\cite{thite2009adaptive} and rotation~\cite{von2021four}, that form spacetimes with different shapes (e.g., cone, prism) or time-variant topology of the spatial mesh.

\textbf{Triangulation in visualization} has focused on nondeforming spacetime so far, primarily for feature tracking problems.  Given a time-invariant spatial discretization, for example, a triangular mesh, one can establish a prismatic mesh connecting all corresponding nodes in spacetime.  One can further tessellate the prismatic mesh with the staircase triangulation scheme~\cite{de2010triangulations, guo2021ftk}, which we review in the next section.

\begin{figure*}
    \centering
    \includegraphics[width=\textwidth]{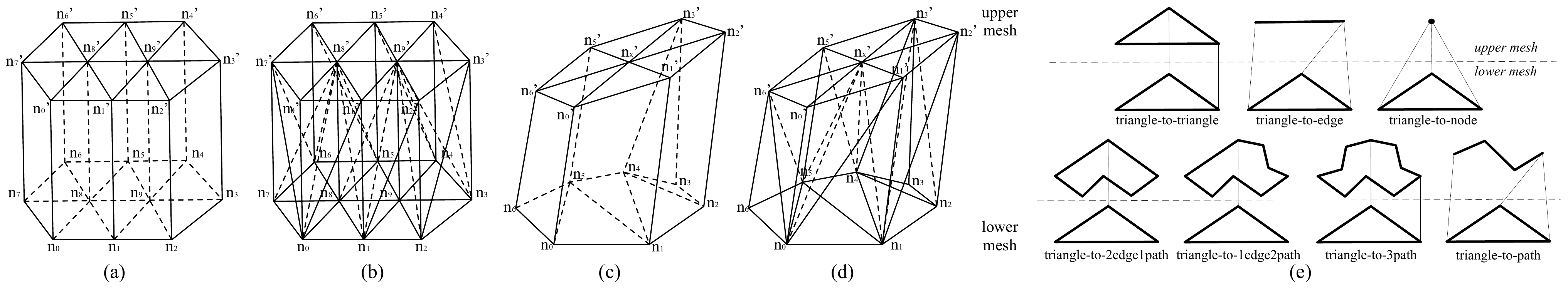}
    \caption{
    Difference between nondeforming and deforming space meshes. (a) An example of nondeforming 3D prismatic mesh obtained by extruding a 2D spatial mesh. (b) A simplicial subdivision of (a). (c) An example of deforming 3D mesh induced by spatially-varying behavior, which has more complex correspondence between the lower and upper meshes: nodes $n_3$ and $n_4$ have the same next node; node $n_x$ has no previous nodes; triangulations in the lower and upper meshes are different; the lower mesh needs translating, rotating, or hybrid transformation to be aligned with the upper mesh. All these differences make it difficult to build a simplicial meshing for the deforming space. (d) A simplicial subdivision of (c). (e) Seven \textit{triangle-to-what} scenarios in deforming 3D spaces. Connected components in lower and upper meshes are linked by next-node correspondence. Adjacent nodes might still be adjacent (in \textit{triangle-to-triangle} scenario), merge (in \textit{triangle-to-edge}, \textit{triangle-to-node}, and \textit{triangle-to-path} scenarios) or become non-adjacent (in \textit{triangle-to-2edge1path}, \textit{triangle-to-1edge2path}, and \textit{triangle-to-3path} scenarios) in the next plane. 
    }
    \label{fig:deforming}
\end{figure*}

\section{Formulation and Preliminaries}
\label{sec:preliminaries}

\begin{figure}[!t]
  \centering
  \includegraphics[width=0.9\columnwidth]{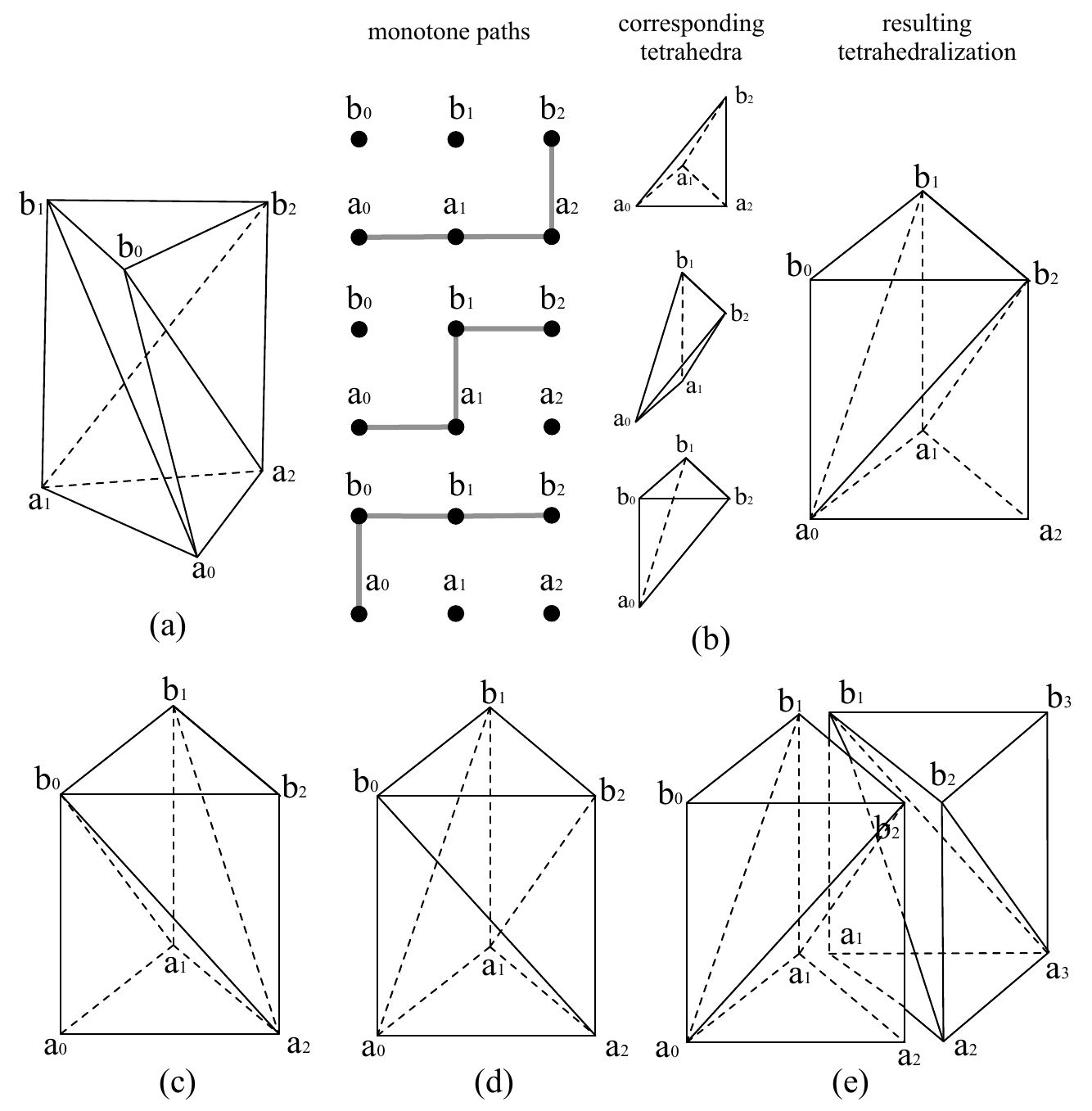}
  \caption{
  Subdivision of a 3D prism. (a) A Sch\"{o}nhardt polyhedron, which is the simplest nonconvex nontriangulable polyhedron. (b) Monotone paths, corresponding tetrahedra, and simplicial subdivisions of a 3D prism by staircase triangulation, with the assumption that the 3 node indices satisfy $a_0<a_1<a_2$. Nodes in the lower mesh are only connected to nodes in upper mesh whose previous nodes have greater indices. (c) A simplicial subdivision of the 3D prism without considering staircase. (d) If all  3 quads are separated by connecting top-left corner and bottom-right corner, then there is no way to triangulate the prism without additional nodes. (e) Two adjacent prisms have conflictive triangulation ways on the common quadrilateral face $a_1-a_2-b_2-b_1$.
  }
  \label{fig:staircase}
\end{figure}

This paper considers the deforming 3D space induced by a magnetic field $\mathbf{B}: \mathbb{R}^3\to\mathbb{R}^3$ (or $\mathbb{R}^3\to\mathbb{R}^2$ if the deformation is invariant over $\phi$) in fusion plasma. Assuming Lipschitz continuity of $\mathbf{B}$, the \emph{deformation} $\Phi: \mathbb{R}^4\to\mathbb{R}^2$ is a continuous function representing the solution of the following initial value problem: 

\begin{equation}
  \frac{\partial\Phi(\mathbf{p}_0, \phi_0, \phi)}{\partial\phi} = \mathbf{B}(\Phi(\mathbf{p}_0, \phi_0, \phi))~\mbox{and}~\Phi(\mathbf{p}_0, \phi_0, \phi_0)=\mathbf{p}_0,
\end{equation}

where $\mathbf{p}_0=(r_0,z_0)^\intercal$ is a 2D spatial location in poloidal plane at $\phi_0$ and $\Phi(\mathbf{p}_0, \phi_0, \phi)$ denotes the $(r,z)^\intercal$ of $\mathbf{p}_0$ after deformation at $\phi$. With the deformation, one can define the field-following interpolation as 
\begin{equation}\label{eq:ff}
  f(\mathbf{p}, \phi) = \beta f(\Phi(\mathbf{p}, \phi, \phi_0)) + (1-\beta) f(\Phi(\mathbf{p}, \phi, \phi_1)),
\end{equation}
where $f: \mathbb{R}^3\to\mathbb{R}$ is a scalar function (e.g., temperature in fusion plasma) and the linear weight $\beta$ is $(\phi_1 - \phi) / (\phi_1 - \phi_0)$, $\phi_0\leq \phi\leq \phi_1$. 

In a discrete sense, we rephrase the objective of the deforming 3D space meshing in Section 1 as follows. Assuming a sequence of $\phi_i$'s with $\phi_0\le \phi_1 \le \ldots \phi_i \le \ldots$ is associated with a 2D simplicial mesh $M_i=\left<V_i, E_i\right>$, where $V_i=\{v_0^{(i)},v_1^{(i)},...\}$ is the set of nodes in $M_i$ and $E_i=\{e_0^{(i)},e_1^{(i)},...\}$ is the set of edges in $M_i$, $v_j^{(i)}$ denotes the $j$-th node and $e_j^{(i)}$ denotes the $j$-th edge in $M_i$. A deforming space mesh is a 3D simplicial mesh $\mathcal{M}=\left<\mathcal{V}, \mathcal{E}\right>$ that contains all nodes $\mathcal{V}=\cup_i V_i$ and edges $\mathcal{E}=\left(\cup_i E_i\right) \cup \left(\cup_i\cup_{j,k} v^{(i)}_j v^{(i+1)}_k \right)$.  The node connectivity $v^{(i)}_j v^{(i+1)}_k$ between $M_i$ and $M_{i+1}$ is defined by the deformation $\Phi$ such that $v^{(i+1)}_k \approx \Phi(v^{(i)}_j, \phi_i, \phi_{i+1})$. We need to consider spatial meshes in two consecutive planes (namely, the \emph{lower} and \emph{upper meshes}), which can be trivially applied to every two adjacent planes for multiple planes. We first review the trivial case, triangulation of nondeforming space ($\mathbf{B}=\mathbf{0}$), and then consider challenges in deforming space.

\subsection{Meshing Nondeforming Space with Staircase Triangulation}
\label{sec:staircase}

With nondeforming space, assuming the underlying spatial mesh $M_i$ is identical for all poloidal planes, the problem is reduced to the following: each node in the lower mesh connects to the same node in the upper mesh.  As a result, each lower-mesh triangle extrudes into a prism, which can be further subdivided into three tetrahedra with \emph{staircase triangulation}~\cite{de2010triangulations}, as illustrated in~\Cref{fig:staircase} (b). As such, staircase triangulation can help mesh 3D nondeforming space~\cite{guo2021ftk} and serve as a basis of our methodology. 

\textbf{Staircase triangulation} provides an arbitrary rule to triangulate a prism based on node indices. Consider a prism $a_0a_1a_2-b_0b_1b_2$ extruded by a triangle $a_0a_1a_2$. With staircase triangulation, each tetrahedron corresponds to a \textit{monotone path} where both alphabets and subscripts are ascending. For example, all monotone paths for prism $a_0a_1a_2-b_0b_1b_2$ are $a_0a_1a_2b_2$, $a_0a_1b_1b_2$, and $a_0b_0b_1b_2$, as shown in~\Cref{fig:staircase}~{(b)} left.  Each path represents a tetrahedron separated from original prism by staircase triangulation (\Cref{fig:staircase}~{(b)} middle), and all tetrahedra corresponding to monotone paths constitute a tetrahedralization of the prism (\Cref{fig:staircase}~{(b)} right).  

\textbf{Meshing the nondeforming space}.  Staircase triangulation directly generalizes to multiple prisms and makes it possible to triangulate prismatic meshes (\Cref{fig:deforming} (a)) extruded from a 2D spatial mesh into 3D (shown in~\Cref{fig:deforming} (b)). 
Note that although staircase triangulation is not the only way to triangulate each prism (see \Cref{fig:staircase} (c) as an example), indices that give a unique ordering of nodes define an elegant and consistent schema to represent a nondeforming space as a simplicial complex. For example, for the two adjacent prisms in~\Cref{fig:staircase}~{(e)}, both prisms triangulate the shared quadrilateral face in a consistent manner by staircase, creating a shared edge $a_1b_2$ that connects the lower index node ($a_1$) to a higher index node ($b_2$); the resulting triangulation includes no conflicting edge, such as $b_2a_1$.  

\subsection{2.5D Representation of Deforming Space: Lower/Upper Meshes and Next Node Connectivity}
\label{sec:space_formulation}

This paper considers the node connectivity defined by \textit{next node} between adjacent meshes; that is, each node in the lower mesh connects to another node in the upper mesh. Next nodes are directly given as output of XGC simulation~\cite{xgc_mesh} and widely used to approximate plasma behavior~\cite{araki2014magnetic,dillow2011enhancing,patel2016design}. For example, in tokamaks or stellarators where the magnetic field defines the deformation, by placing a particle at a mesh node $v$ in a lower mesh, scientists can build a connection between $v$ and the node $v^\prime$ (called \textit{next node} of $v$) closest to the particle's position in the upper mesh (\Cref{fig:teaser} (b)).

\textbf{2.5D independent partition}. Denote the lower and upper spatial meshes by $S$ and $S^\prime$, respectively. We define a \emph{2.5D independent partition} $PQ$, where $P$ and $Q$, respectively, is a simply connected component in the lower and upper meshes such that (1) any node $v\in P$ maps to another node $v^\prime\in Q$ and (2) any node $v\in \partial P\setminus\partial S$ maps to another node $v^\prime\in \partial Q\setminus\partial S^\prime$; the symbol $\partial$ denotes the boundary of a partition.  Note that $P$ or $Q$ may be degenerate, such as a single node or one or multiple edges.  For example, the entire lower and upper meshes are a 2.5D independent partition. For another example, assuming $P$ contains one single triangle, the next-node mapping implicitly determines diverse situations on how $P$ maps to upper-mesh elements, which could be a single node, one or multiple edges, or one or multiple triangles as finite scenarios illustrated in~\Cref{fig:deforming} (e). 
The seven scenarios indicate different local deformation behaviors in space. If we consider the deformation in the scale of edge, then two adjacent nodes may still be mapped to two adjacent nodes (small change of relative distance between the two nodes), merge to one node (shrinking behaviors), or break into a path (expanding behaviors) in the upper mesh.  

\textbf{2.5D domain decomposition}.  A decomposition of $PQ=\cup_i P_iQ_i$ is defined by nonoverlapping 2.5D independent partitions $P_iQ_i$.  Formally, for any two different partitions $P_iQ_i$ and $P_jQ_j$, the intersection is either null or lower-dimensional simplicies like nodes and edges.  The union and intersection of two 2.5D independent partitions $P_iQ_i$ and $P_jQ_j$ are defined by $(P_i\cup P_j)(Q_i\cup Q_j)$ and $(P_i\cap P_j)(Q_i\cap Q_j)$, respectively.  We will further discuss how the 2.5D domain is decomposed into smaller partitions for divide-and-conquer processing in the next section.

\textbf{Triangulation of a 2.5D independent partition}.  A 2.5D independent partition $PQ$ consists of only triangles, edges, and nodes from 2D complexes $P$ and $Q$ and does not bound a 3D volume; one must first transform $PQ$ into a proper polyhedron before triangulating into tetrahedra.  We refer to deriving a polyhedron from a 2.5D independent partition as \emph{lateral triangulation}. To form a polyhedron, because $P$ and $Q$ are already triangles or degenerate triangles, one has to define a surface/triangular mesh that connects 2D boundaries $\partial P$ and $\partial Q$; the lateral surface must contain edges between a node $v\in \partial P$ and its next node $v^\prime\in \partial Q$. However, multiple choices exist to find a lateral triangulation, as two ways exist to triangulate a quadrilateral face on a prism.  The non-uniqueness leads to challenges: (1) the resulting polyhedron may be nontriangulable, as discussed below, and (2) two 2.5D independent partitions should have the same triangulation on their common lateral face, or otherwise one polyhedron will overlap or form a void with its neighboring polyhedron, as addressed in the next section.

\begin{wrapfigure}{LH}{0.74\linewidth}
    \centering
    \includegraphics[width=\linewidth]{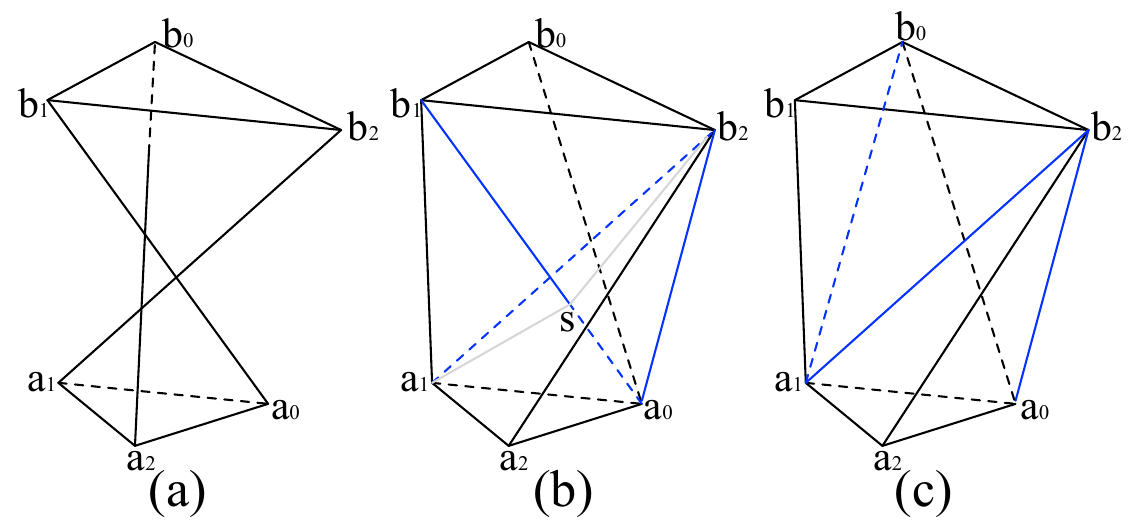}
    \caption{Example of ill-posed and quasi-ill-posed prisms. (a) An ill-posed prism. (b) Staircase triangulation makes the deformed prism nontriangulable without Steiner point $s$, because newly added line segment $a_0b_1$ intersects with newly created triangle $a_1a_2b_2$ at point $s$. Visible faces of the polyhedron include $a_1b_1s$, $b_1b_2s$, $a_1a_2b_2s$, and $a_0a_2b_2$. (c) The deformed prism is successfully triangulated after replacing $a_0b_1$ by $a_1b_0$ to separate 4-node cycle $a_0a_1b_1b_0$.}
    \label{fig:fake_ill}
\end{wrapfigure}

\textbf{Handling of nontriangulable 2.5D independent partition}.  Some partitions cannot find a triangulation, which we call \emph{ill-posed partitions}.  For example, the prism in~\Cref{fig:fake_ill} (a) is twisted so much that any lateral triangulation will lead to self-intersections.  Although the prism is a triangle-to-triangle case, the ill-posed prism does not legitimately enclose a volume. Ill-posed partitions are separate from  \emph{quasi-ill-posed partitions}, which can triangulate with a different lateral triangulation.  We show in Supplemental Materials that a prism has eight possible ways to tessellate three quadrilateral faces; two of them lead to a Sch\"{o}nhardt prism and cannot triangulate, but choosing a different scheme will avoid this situation (\Cref{fig:fake_ill} (b) and (c)). We further discuss the adjustment of lateral triangulation later.

Should an ill-posed partition appear while no alternative way exists to decompose the domain, the space cannot be triangulated. Unfortunately, ill-posed partitions suggest that poloidal discretization is insufficient. In this case, one can refine the poloidal resolution with a smaller gap $\Delta \phi$ so that no ill-posed partitions appear.\footnotemark~In other words, one can upsample the toroidal resolution with field-following interpolation (Eq.~\eqref{eq:ff}) before space meshing.
\footnotetext{In general, the distortion decreases with a finer poloidal resolution $\Delta \phi$; the distortion characterized by displacement $||\Phi(\mathbf{p}, \phi_0, \phi) - \mathbf{p}||$ is bounded by $L\cdot \Delta \phi$, $L$ being the Lipschitz constant of $\mathbf{B}$.} 

\begin{figure*}[!t]
    \centering
    \includegraphics[width=\textwidth]{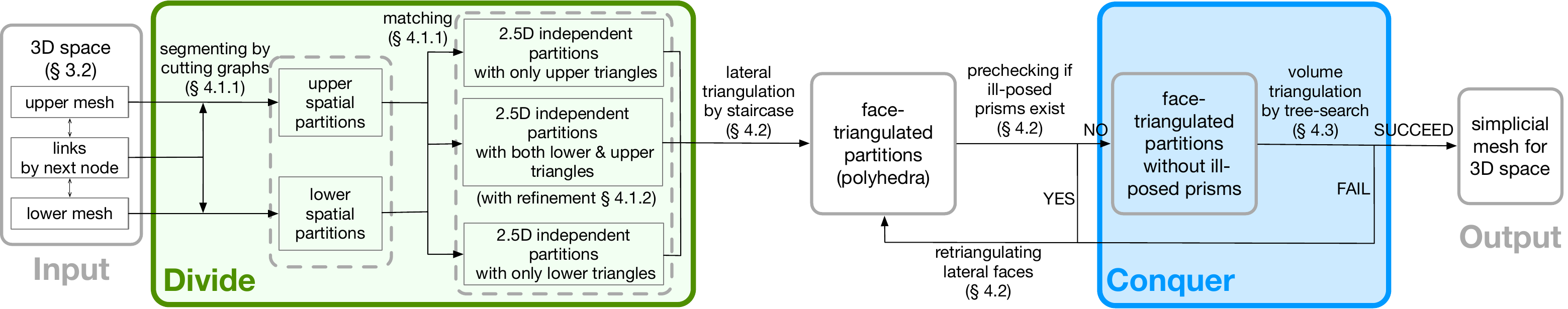}
    \caption{Pipeline of our tetrahedralization algorithm. The algorithm goes through two stages, \textit{divide} and \textit{conquer}. In the divide stage, the input is a deforming 3D space formed by lower mesh, upper mesh, and next node connectivities between them, and the output is a set of 2.5D independent partitions. Then we triangulate their lateral faces. The conquer stage takes every face-triangulated polyhedron as an input. Note that the original 3D space is also a 2.5D independent partition and one can skip the divide stage without complexity consideration. Indivisible prisms are filtered before dividing them, while more complex polyhedra need tentatively dividing first and retriangulating lateral faces if they fail to be divided. The output of the conquer stage, as well as the whole algorithm, is a set of tetrahedra that form a simplicial mesh of the input 3D space.}
    \label{fig:pipeline}
\end{figure*}

\section{Methodology}
\label{sec:method}

\Cref{fig:pipeline} illustrates our workflow to triangulate a deforming 3D space with a \textit{divide-and-conquer} strategy, with the rationale that the domain may be decomposed into 2.5D independent partitions (or simply partitions when there is no ambiguity), each of which could lead to a polyhedron for further triangulation. Multitude challenges exist, including (1) defining a proper domain decomposition that could lead to triangulatable partitions, (2) maintaining compatibility, that is, ensuring no overlaps or voids exist between neighboring partitions, (3) finding a proper lateral triangulation scheme that leads to a possible 3D triangulation, and (4) triangulating a nonconvex polyhedron.  Each challenge is nontrivial and requires careful design.  To these ends, we design our methodology as three major phases with multiple refinements and trial-and-error steps: 
\begin{itemize}\setlength{\parskip}{0pt}
  \item \textbf{Space decomposition} (divide stage): Decomposing the domain into as many partitions as possible based on \emph{closed cutting paths} and inner-node refinements (\Cref{sec:divide}).
  \item \textbf{Lateral triangulation}: For each partition, forming a polyhedron by meshing the surface that connects the partition's lower and upper boundaries.  By default, staircase triangulation rules are applied when possible but subject to change if later volume triangulation is impossible (\Cref{sec:ftp}).
  \item \textbf{Volume triangulation} (conquer stage): For each polyhedron, searching for a possible triangulation scheme with a two-tiered decision-tree algorithm (\Cref{sec:conquer}).
\end{itemize}

\subsection{Space Decomposition}
\label{sec:divide}

With the given lower and upper meshes and node connectivities between them, the first step is partitioning the space into nonoverlapping 2.5D independent partitions.  We aim to decompose the domain into much smaller partitions so that we do not need to triangulate the entire space with our decision-tree algorithm.

We introduce the definition of \emph{cutting edges}, \emph{cutting graphs}, and \emph{cutting paths} and prove that the domain can be decomposed into 2.5D independent partitions (as defined in~\Cref{sec:space_formulation}) as connected components isolated by cutting graphs.

\textbf{Cutting edges}.  Formally, with a partition $PQ$, edges $uv\in P$ and $u^\prime v^\prime$ form a pair of cutting edges if there exists edge $u^\prime v^\prime\in Q$, where $u^\prime$ and $v^\prime$ are next nodes of $u$ and $v$, respectively. As shown in~\Cref{fig:divide}~{(a)}, a pair of cutting edges is a 4-node cycle formed by four edges: a lower-mesh edge $n_1n_9$, an upper-mesh edge $n_1^\prime n_9^\prime$, $n_1n_1^\prime$, and $n_9n_9^\prime$.  
\begin{restatable}{assumption}{myassumption}\label{myassumption}
The lower/upper mesh and toroidal resolution are sufficiently refined such that for any point $x$ on a cutting edge $v_0v_1$ in the lower mesh, the nearest node of $\Phi(x)$ on the upper mesh is either $v_0^\prime$ or $v_1^\prime$.
\end{restatable}
Here, we omit $\phi$ in $\Phi$ for clarity. Conceptually, we assume that a cutting edge reflects distinct transport behaviors of the magnetic field $\mathbf{B}$, which drives the deformation. This is analogous to a laminar flow parallel to the cutting edge, where no particles cross the edge. 

\textbf{Cutting graphs and cutting paths}. A cutting graph is a graph that connects nodes that are connected by cutting edges.  A cutting path is an arbitrary path in the cutting graph.  For example, path $u_0u_1u_2u_3$ is a cutting path if $u_0u_1$, $u_1u_2$, and $u_2u_3$ are cutting edges; the counterpart consisting of next nodes $u_0^\prime u_1^\prime u_2^\prime u_3^\prime$ is also referred to a cutting path in the upper mesh.  

\begin{restatable}{lem}{mylemma}\label{mylemma}
Partition $PQ$ is independent if the boundary $\partial P$ is a closed cutting path.
\end{restatable}

\begin{proof}
In a trivial case, if all nodes on $P$ are on boundary $\partial P$, $PQ$ is independent because every node on $P$ maps to $Q$.  Now, assume $v_k$ is an inner node of $P$ such that $v_k\notin\partial P$ and the next node is $v_k^\prime\notin Q$.  Because the next node is the closest point to $\Phi(v_k)$, one can find a sufficiently small disk $B(v_k, \epsilon), \epsilon>0$ so that the deformation $\Phi$ of any point on the disk has the closest point $v_k^\prime$.  Now, the domain $P$ maps to at least two connected components by $\Phi$; one is the image of the disk, and another is the region that intersects $Q$.  Having two connected components contradicts Rudin's Theorem 4.22~\cite{rudin}; the continuous mapping $\Phi$ of $P$ shall be connected. This terminates the proof.
\end{proof}

The assumption and lemma make it possible to decompose the domain via connected component labeling.  Specifically, cutting graphs break the lower and upper meshes into nonoverlapping 2.5D independent partitions for further processing. We use a union-find implementation to join edge-sharing triangles in each component and then match the components between the lower and upper meshes.  Special treatment may be needed for degenerate cases on domain boundaries; for example, if the cutting path in the upper mesh is on the boundary, a connected component in the lower mesh may not find a matching component in the upper mesh. In this case the resulting partition is degenerate, and we assign a path as the partition's upper component, as illustrated in~\Cref{fig:divide}.  

Decomposing the domain by cutting graphs provides several benefits. First, a closed cutting path helps isolate trivial scenarios such as deformed prisms  (i.e., cells with triangle-to-triangle scenario in~\Cref{sec:space_formulation}) from more complex scenarios, because all quadrilaterals in deformed prisms are 4-node cycles formed by two cutting edges. 
Second, the triangulation of a 4-node cycle of a cutting edge can be directly given by staircase triangulation as introduced in~\Cref{sec:staircase}, which makes it easy to coordinate cross sections. 
Third, each partition resulting from our decomposition could be processed independently, according to Lemma 1.

\begin{figure}[!t]
    \centering
    \includegraphics[width=\columnwidth]{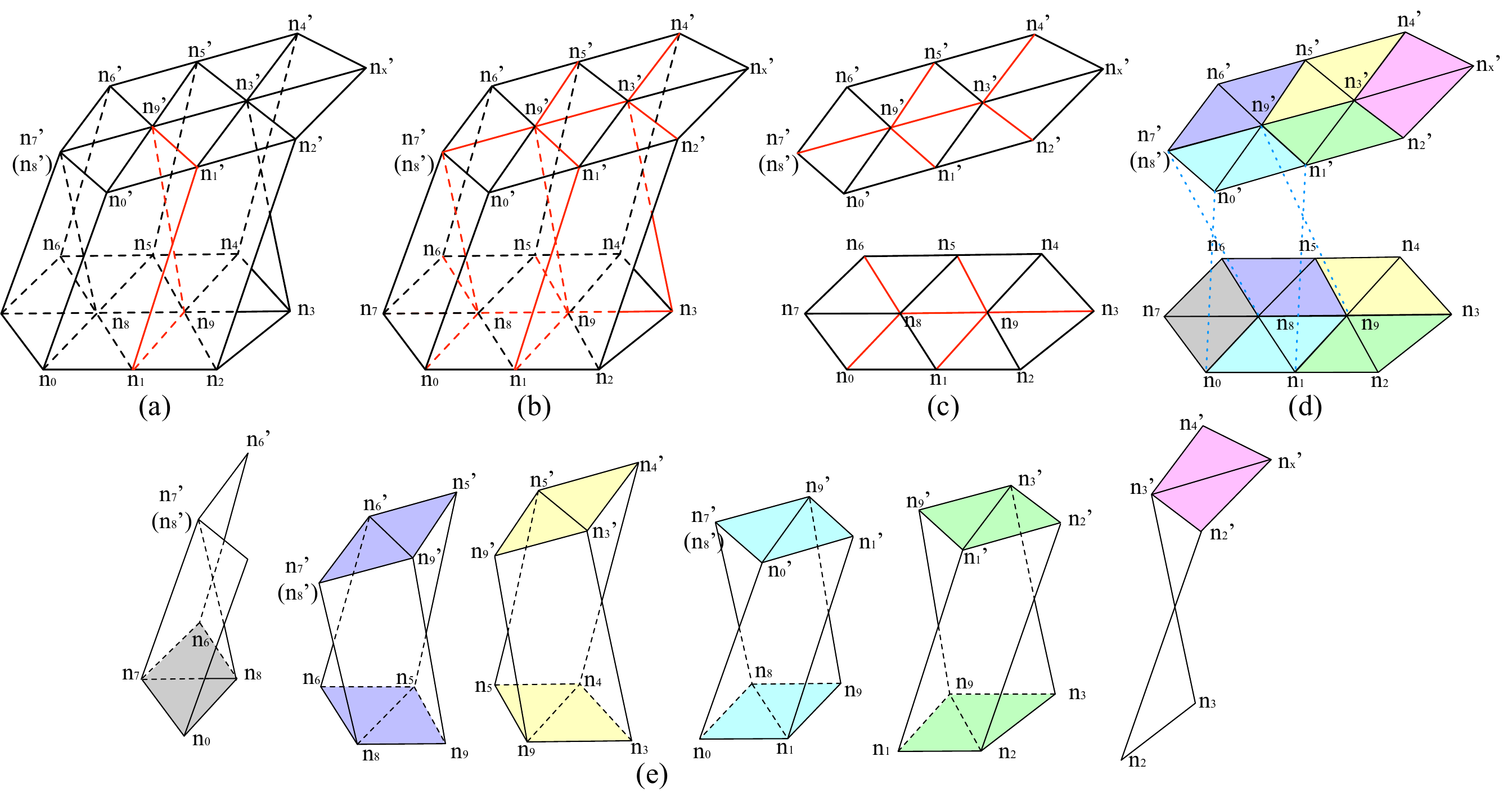}
    \caption{Mini-example showing the process of segmenting space by cutting paths in the divide stage. (a) 4-node cycle formed by two cutting edges highlighted in red in the space. The cycle may or may not be a plane. (b) All the cutting paths in the space. (c) Segmenting lower and upper meshes into spatial partitions by cutting paths. (d) Matching lower and upper spatial partitions by node-to-next-node correspondences. Only correspondences between the blue partitions are shown. Some partitions are matched with degenerate spatial partitions (e.g., edges or nodes), such as the gray partition in the lower mesh and the pink partition in the upper mesh. (e) After matching the spatial partitions, we add node-to-next-node correspondences back and get 2.5D independent partitions.}
    \label{fig:divide}
\end{figure}

\begin{figure}[!t]
    \centering
    \includegraphics[width=\columnwidth]{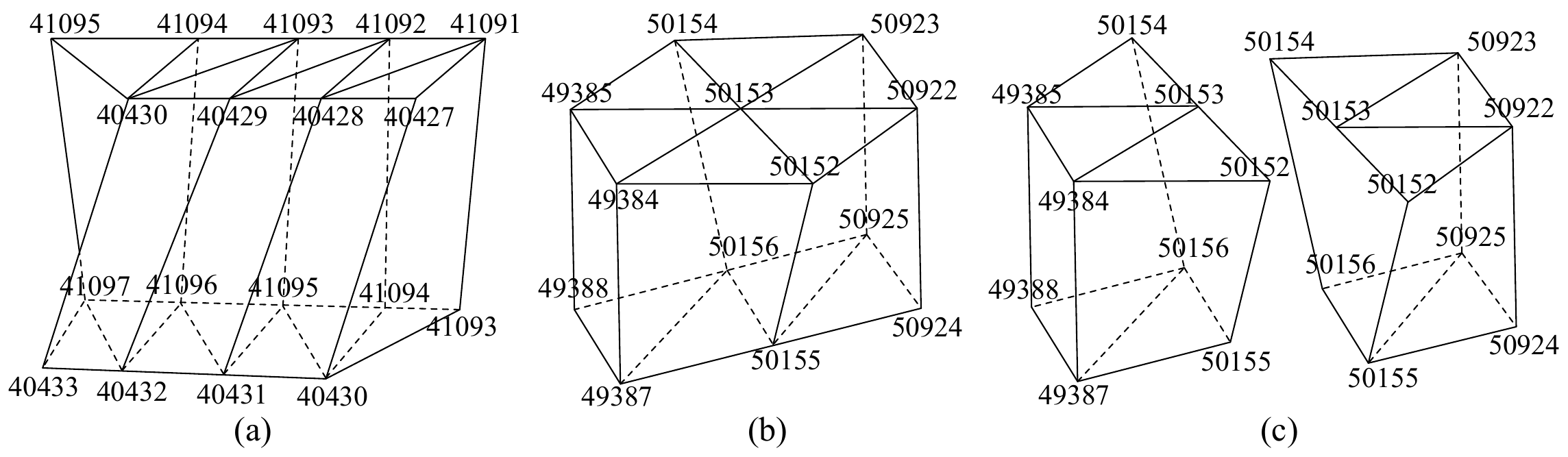}
    \caption{(a) and (b) show two 2.5D independent partitions derived from ITER data. The indices of mesh nodes in all figures agree with those in ITER data. (a) has no inner nodes in upper mesh, while (b) has one inner node 50153 on upper mesh. (c) shows two partitions of (b) derived by optional refinement in \textit{divide} step.
    }
    \label{fig:indSys}
\end{figure}

\textbf{Refining a Partition by Inner Nodes}. Readers may skip this step if the 2.5D independent partitions are small enough to reduce complexity. Some 2.5D independent partitions isolated by cutting graphs could still include a large number of triangles or even inner nodes in the lower or upper mesh. To derive smaller partitions and further reduce complexity in the latter processes, we decompose every 2.5D independent partition with inner node(s). These partitions are split by \textit{penta-faces} that include an edge in the lower mesh, two edges between two nodes and their next nodes, and a path with two edges in the upper mesh passing the inner node (e.g., $50155-50156-50154-50153-50152$ in \Cref{fig:indSys}~{(b)}). \Cref{fig:indSys}~{(c)} shows the split result of the partition in \Cref{fig:indSys}~{(b)}. There might be more than one penta-face crossing one inner node; we use the one that can balance the number of cells in both the lower and upper meshes in two split partitions. Note that one can further refine a 2.5D independent partition with multiple adjacent inner nodes by multinode cycles, not only by penta-faces with 5-node cycles. Splitting by multinode cycles yields smaller partitions but makes triangulation and compatibility on cross sections more complex.

\subsection{Lateral Triangulation}
\label{sec:ftp}

\begin{wrapfigure}{LH}{0.45\linewidth}
    \centering
    \includegraphics[width=\linewidth]{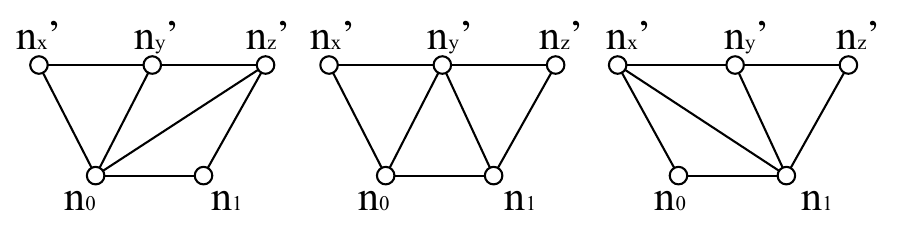}
    \caption{All possible lateral triangulations for a penta-face. Each of them can be used as long as the subdivisions of two partitions containing the same penta-face agree on this common penta-face. Cutting by faces with more nodes in a cycle is feasible but incurs more complexity of subdividing and coordinating cross sections. For a convex polygon with $n_p$ nodes, there are $\frac{1}{n_p-1}{2n_p-4 \choose n_p-2}$  triangulations for it~\cite{de2010triangulations}.}
    \label{fig:penta-face}
\end{wrapfigure}

\label{sec:face_triangulation}

Before triangulating a 2.5D independent partition, one must first triangulate its lateral faces to form a polyhedron, namely, a \textit{face-triangulated partition}.  Note that components from the lower/upper mesh are already simplicial; thus, we need to consider only lateral faces.  The lateral faces created in the divide stage include 4-node and 5-node cycles.  There are two subdivisions for 4-node cycles and three subdivisions for 5-node cycles (shown in~\Cref{fig:penta-face}).  The subdivision of a lateral face can be arbitrary if the subdivisions of two adjacent partitions agree on it.  As stated in~\Cref{sec:staircase}, we follow the staircase triangulation rule---choosing the splitting edge with ascending node indices---to coordinate the separation of cross sections caused by cutting edges in the divide stage.  However, triangulation in deforming space has more problems and possibly needs more than one iteration.

\textbf{Lateral triangulation of deformed prisms}. As we discussed in~\Cref{sec:preliminaries}, not all deformed prisms can be triangulated, which we refer to as \textit{ill-posed prisms} in this paper.  Take a step back to consider a regular prism with eight potential ways to triangulate its lateral quadrilaterals, as illustrated in~\Cref{sec:staircase}; six triangulation schemes can successfully split the prism into three tetrahedra, except for the two schemes that add only opposite diagonals.  Note that staircase triangulation is a particular case of the six successful schemes. Now consider a deformed prism; not all of the six separations can triangulate a deformed prism.  Some deformed prisms (called \textit{quasi-ill-posed prisms}) 
fail to triangulate by the staircase triangulation but are successfully triangulated by another triangulation scheme (see \Cref{fig:fake_ill} (b) and (c) as an example). Unfortunately, ill-posed prisms cannot be triangulated by any scheme, as shown in Supplemental Materials.

We follow three steps to triangulate the lateral faces of a deformed prism.   
First, we check whether staircase triangulation could lead to a valid triangulation.  
If not, second, the deformed prism is at least a quasi-ill-posed prism, which requires checking all separation schemes to see whether one of the separations can triangulate the deformed prism. Third, if all separations cannot triangulate the deformed prism, it is an ill-posed prism. It also means that the toroidal resolution of data should be increased to reduce the deformation between two successive poloidal planes, as we previously discussed in~\Cref{sec:space_formulation}. See Supplemental Materials for more details about conditions to check.

\textbf{Lateral triangulation of complex partitions beyond prisms}. We choose an arbitrary sequence of separations of the 2.5D independent partition's lateral faces that are not constrained by triangulation of quasi-ill-posed prisms, and we try to triangulate it. If we fail to triangulate the 2.5D independent partition with given triangulated faces, then we change the separation of one lateral face and try to triangulate the newly created polyhedron, as shown in~\Cref{fig:pipeline}. Different from deformed prisms, it is hard to determine whether a polyhedron is triangulable without trying to triangulate it first.

\subsection{Volume Triangulation}
\label{sec:conquer}

Once each 2.5D independent partition is face-triangulated into a polyhedron, we further triangulate each polyhedron into tetrahedra.  We propose a \textit{node elimination} algorithm that removes nodes one by one as well as edges incident on it as tetrahedra until all nodes are removed.  In searching for the sequence of nodes to be removed, there are multiple decisions to make, and we need to trace back to previous steps if we meet with nontriangulable remaining polyhedron. Thus, we formulate the node elimination algorithm into a decision-tree paradigm.

\textbf{Two-tiered decision tree for volume triangulation}. The decision tree makes two types of decisions alternately in levels with odd and even numbers (odd and even levels). In odd levels, the tree makes a decision on which node to eliminate (so-called \textit{pivot node}); in even levels, the tree chooses one way to triangulate the polyhedron formed by pivot node and edges incident on the node (so-called \textit{isolated polyhedron}). Every tree node contains information about separated tetrahedra, remaining polyhedron, pivot node, and unvisited children of the tree node.

\begin{wrapfigure}{LH}{0.6\linewidth}
    \centering
    \includegraphics[width=\linewidth]{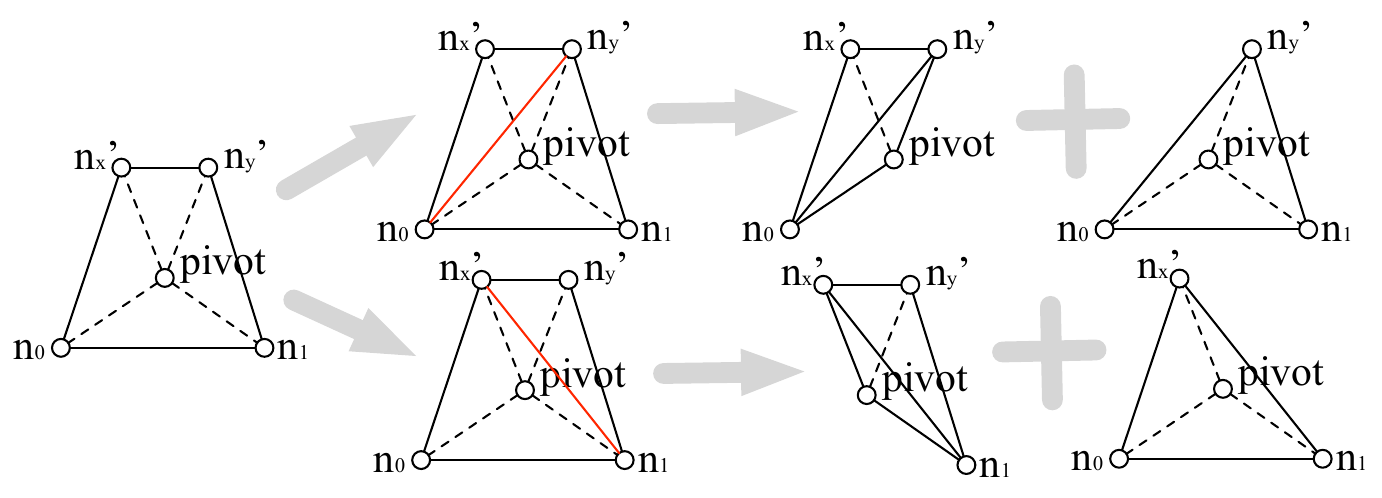}
    \caption{One example of isolated polyhedron consisting of a pivot node on lower mesh, its neighbors, and their connectivities (left), two possible triangulations (middle), and separated tetrahedra (right).}
    \label{fig:pivot_cases}
\end{wrapfigure}

\textbf{Odd levels: decision on choosing a pivot node}. The pivot node can be any node in the polyhedron, and one can introduce any heuristic to choose pivot node. Here, we prioritize all mesh nodes with minimum degree in the current remaining polyhedron as pivot nodes and create a tree node for every pivot node as child for current tree node. One advantage of this heuristic is that nodes with minimum degree are incident to fewer edges and thus incur simpler isolated polyhedra. See~\Cref{alg:pivot_choose} for pseudo-code of this step.

\textbf{Even levels: decision on triangulating polyhedron around the pivot node}. We continue with a pivot node chosen in the last level and eliminate the pivot. We first isolate the pivot, its neighbors, and edges between the pivot and its neighbors. All the isolated elements form an isolated polyhedron that is a subset of the inputted polyhedron. Then we triangulate the isolated polyhedron by exhausting all possibilities for an isolated polyhedron, remove derived tetrahedra from the current polyhedron, and update the remaining polyhedron. All possible tetrahedralizations of an example isolated polyhedron are shown in~\Cref{fig:pivot_cases}. The isolated polyhedron typically contains 4–5 nodes. A complete list of isolated polyhedra with no more than 5 nodes, along with their corresponding tetrahedralizations, is provided in the Supplemental Materials. See~\Cref{alg:pivot_remove} for the pseudo-code of this step.

Note that not all the tetrahedralization ways in the complete list of~\Cref{fig:pivot_cases} are feasible for a polyhedron; two conditions need checking: (1) if every newly created link is inside of the polyhedron (constraints introduced by nonconvexity) and (2) if any of separated tetrahedra has four coplanar mesh nodes. We create a tree node for every feasible tetrahedralization as a child for the tree node representing the pivot node in the upper level, and we update separated tetrahedra and remaining polyhedron for these children. In the next iteration, we repeat operations in the two levels for the remaining polyhedron while the remaining polyhedron is non-empty. \Cref{fig:node-elimination} shows the node elimination algorithm in the structure of the decision tree, taking an arbitrary polyhedron with triangulated lateral faces as an example. See~\Cref{alg:pivot_remove} for the pseudo-code of the whole node elimination algorithm.

\begin{figure}[!ht]
    \centering
    \includegraphics[width=0.458\textwidth]{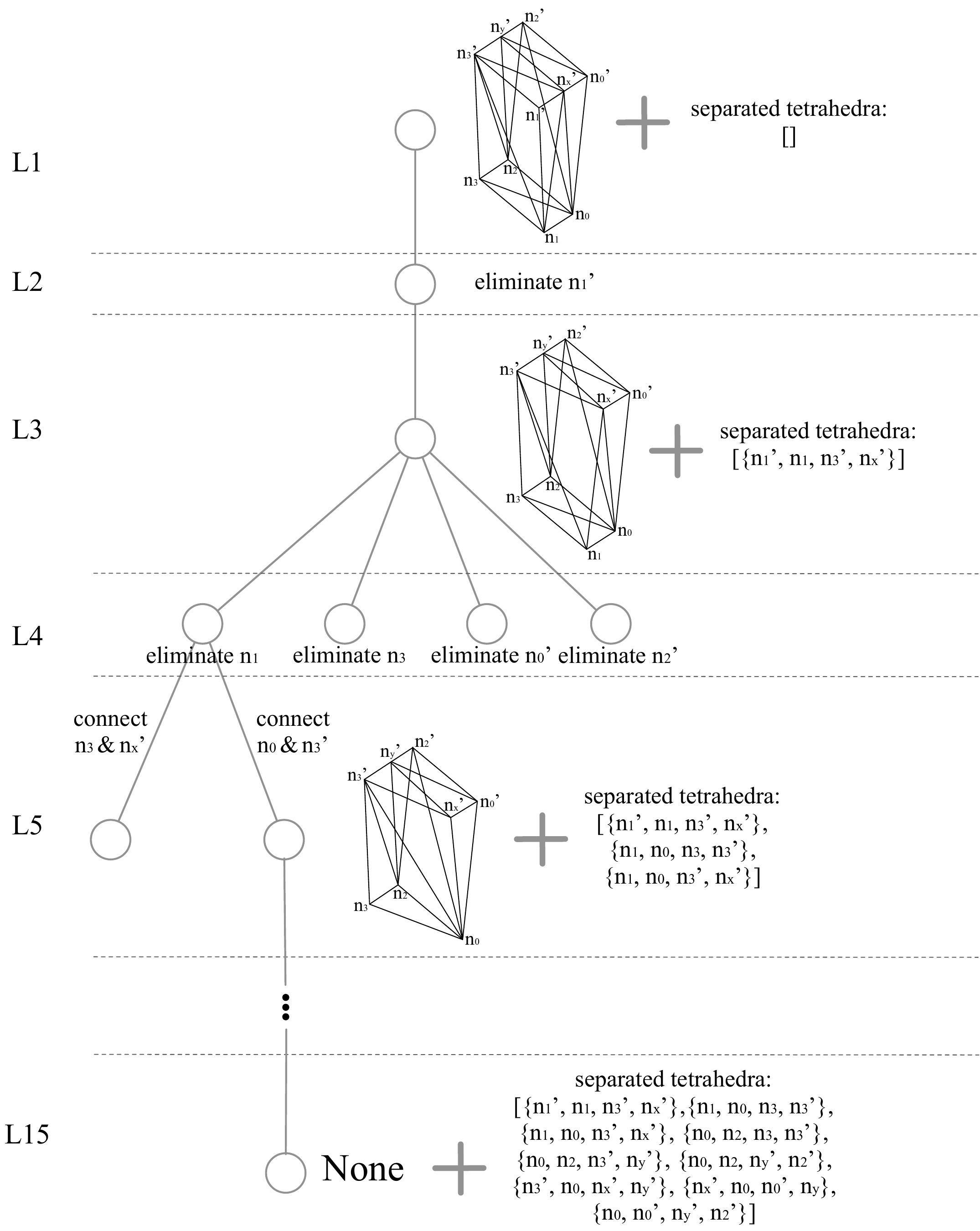}
    \caption{Node elimination algorithm in a decision-tree paradigm. In odd levels, it makes a decision on ``which mesh node to eliminate'' among mesh nodes with minimum degree; in even levels, it makes a decision on ``how to eliminate'' based on the isolated polyhedron formed by the pivot node and its neighbors. See the full tree in Supplemental Materials.
    }
    \label{fig:node-elimination}
\end{figure}

\textbf{Complexity of the algorithm}. The best and most common case is to terminate with a single branch ($O(n)$); the worst and nearly impossible case is to fully expand the tree ($O(n2^n)$). This tree has $2(n_v-3)+1$ levels ($n_v$ stands for the number of nodes in the face-triangulated polyhedron to be triangulated) because nearly every node (except the last three nodes that can be directly removed with the last tetrahedron) will be a pivot exactly once in any branch that connecting root and a leaf, and every pivot generates two levels. To find a tetrahedralization way, there is no need to traverse all possible tree node. The algorithm terminates when it finds a tree node whose remaining polyhedron is empty. When it meets with an indivisible isolated polyhedron (e.g., \Cref{fig:staircase}~{(a)}) around a pivot node, it traces back and continues with the next unvisited child of the current tree node's parent.

\textbf{Completeness of the algorithm}. The algorithm is guaranteed to find a volume triangulation that allows only connections between nodes and their 2-hop neighbors if there is one, because this decision tree covers all possible connections within 2-hop neighbors. If there are no unvisited nodes in the tree, the algorithm terminates with no solution, which means that the inputted polyhedron is indivisible if only links between 2-hop neighbors are allowed to be created. It is reasonable to consider links between 2-hop neighbors because links between greater than n-hop neighbors degrade the accuracy of interpolation in tetrahedra.

\textbf{Optimality of the algorithm.} If cost of a solution in the algorithm is defined to be the number of visited nodes when finding a feasible volume triangulation, then our heuristics for choosing pivot may not direct the algorithm to an optimal solution. Choosing mesh nodes with minimum degree as a pivot may not be the most complexity-saving heuristics: removing a mesh node with minimum degree might introduce nonconvexity to the remaining polyhedron while removing mesh nodes with higher degrees could maintain convexity. Nonconvexity brings more constraints on feasible connections that one can create and raises the risk of making the remaining polyhedron indivisible, which requires tracing-back in the decision tree and thus increases cost.

\begin{algorithm}[!ht]
\small
\caption{Node Elimination Algorithm}\label{alg:node_elimination}
\KwData{a face-triangulated polyhedron}
\KwResult{tetrahedralization result for the polyhedron}
root $\gets$ Node(\textit{separated\_tets} = $\emptyset$, \textit{remaining\_volume} = volume,
 \textit{pivot} = \texttt{None}, \textit{unvisited\_children} = $\emptyset$)\;
tree $\gets$ Tree(root=root)\;
current\_node $\gets$\ root\;
divisible $\gets$ \texttt{True}\;
\While{current\_node.remaining\_volume$\neq\emptyset$}{
\If{divisible == \texttt{False}}{
    current\_node $\gets$ the next unvisited child of current\_node's parent\;
    \If{no unvisited nodes in the tree}{ \textit{Throw}(``This polyhedron is indivisible.'')
}
\eIf{current\_node.current\_node.pivot==\texttt{None}}{
    current\_node, tree, divisible $\gets$ \textit{PivotChoose}(current\_node, tree)\;
    }{
    current\_node, tree, divisible $\gets$ \textit{PivotRemove}(current\_node, tree)\;
    }
}}
\textbf{return} current\_node.\textit{separated\_tets}
\end{algorithm}

\begin{algorithm}[!ht]
\small
\caption{PivotChoose}\label{alg:pivot_choose}
\KwData{tree\_node, tree}
find mesh nodes with minimum degree\;
generate children representing pivot nodes for tree\_node\;
divisible $\gets$ \texttt{False} \;
\While{divisible==\texttt{False}}{
\eIf{tree\_node.unvisited\_children$\neq\emptyset$}{
current\_node $\gets$ tree\_node.\textit{unvisited\_children}[0]\;
remove current\_node from tree\_node.\textit{unvisited\_children}\;
current\_node, tree, divisible $\gets$ \textit{PivotRemove}(current\_node, tree)\;
}{
\textbf{return} current\_node, tree, \texttt{False}\;
}
\textbf{return} current\_node, tree, \texttt{True}\;
}
\end{algorithm}

\begin{algorithm}[!ht]
\small
\caption{PivotRemove}\label{alg:pivot_remove}
\KwData{tree\_node, tree}
check the cycles formed by pivot and its neighbors; determine its type\;
divisible $\gets$ if at least one divisible way exists\;
\eIf{divisible}{
generate children representing dividing ways for tree\_node\;
current\_node $\gets$ the first unvisited child of tree\_node\;
remove current\_node from tree\_node.\textit{unvisited\_children}\;
update current\_node.\textit{separated\_tets} and current\_node.\textit{remaining\_volume}\;
return current\_node, tree, \texttt{True}
}{
return current\_node, tree, \texttt{False}
}

\end{algorithm}

\section{Results and Evaluation}
\label{sec:case-studies}

We evaluate our method by XGC simulations that discretize the domains of ITER and W7-X with 16 and 1,280 polodial planes, respectively. Specifically, for W7-X data, 256 polodial planes repeat five times over $\phi\in[0,2\pi)$. Each triangular mesh contains $O(10^5-10^6)$ nodes and triangles. The triangular mesh is identical across all poloidal planes in ITER, while W7-X features different meshes on the 256 poloidal planes. All experiments were conducted on a machine with an Apple M1 chip and 64 GB of system memory, without parallel technology.

\subsection{Performance Evaluation of Our Method}

We compare our method with TetGen~\cite{hang2015tetgen}, a general-purpose tool for generating tetrahedral meshes from a 3D input domain, with results shown in~\Cref{tab:tetgen}. The difference in purpose and algorithm between our method and TetGen is discussed in~\Cref{sec:discussion}. TetGen only accepts constraints defined by cycles (e.g., triangles), so we specify the next node connectivity constraints between adjacent meshes using the cycles formed by pairs of cutting edges (\Cref{sec:divide}). Our method generates 339,626 tetrahedra with $56,980\times2=113,960$ nodes for ITER data, and $1,117,419$ to $1,121,074$ tetrahedra with $187,561\times 2=375,122$ nodes for every two adjacent poloidal planes in W7-X. TetGen is unable to generate a valid tetrahedralization that satisfies the constraints imposed by the two input meshes and the next-node connectivity for ITER and W7-X data. To enable a fair comparison, we generate synthetic data, with details provided in the Supplementary Materials. Across all datasets, our method is at least 10 times faster than TetGen. Notably, our approach avoids the introduction of unnecessary Steiner points, which reduces the time required to compute scalar values at these points through magnetic line-following interpolation.

\begin{table}[!th]
    \centering
    \caption{Performance comparison between our method (based on C++ implementation) and TetGen. For W7-X, we aggregate the number of nodes, triangles, and timings for 256 pairs of adjacent poloidal planes. For ITER and synthetic data with identical mesh across poloidal planes, results are shown for a single pair of poloidal plans.}
    \small
    \begin{tabular}{p{0.6cm}|c|c|c|c|p{0.2cm}|c}
    \toprule
         Dataset & \# Nodes & \# Tris & Method & Time (sec) & Suc. & \# Steiner \\ \midrule
         \multirow{2}*{ITER} & \multirow{2}*{56,980} & \multirow{2}*{112,655} & ours & 0.45659 & Yes & 0 \\ \cline{4-7}
         & & & TetGen & 5.10149 & No & - \\ \hline
         \multirow{2}*{W7-X} & \multirow{2}*{48,015,616} & \multirow{2}*{95,390,720}  & ours & 371.226 & Yes & 0 \\ \cline{4-7}
         & & & TetGen & 6179.84 & No & - \\ \hline
         \multirow{2}*{\makecell{Syn.\\data}} & \multirow{2}*{1,000} & \multirow{2}*{1,903} & ours & 0.00584 & Yes & 0 \\ \cline{4-7}
         & & & TetGen & 0.35879 & Yes & 65 \\
    \bottomrule
    \end{tabular}
    \label{tab:tetgen}
\end{table}

\subsection{Evaluation on Visualization Tasks}

\begin{figure}
    \centering
    \includegraphics[width=0.76\columnwidth]{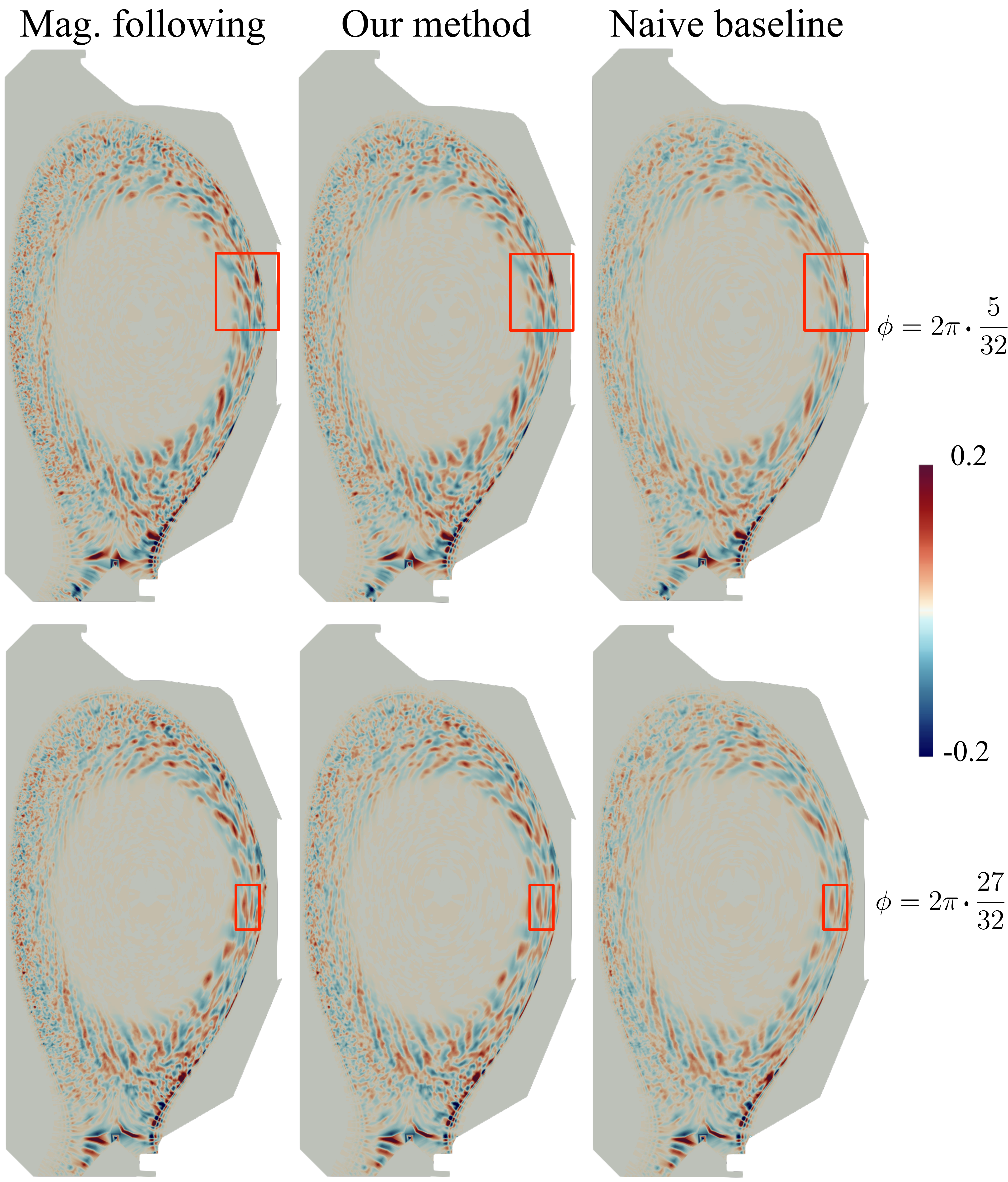}
    \caption{Upsampling results of magnetic-following interpolation, our method, and a naive baseline. The overall patterns shown by different interpolation methods are similar. However, our method has closer values to magnetic-following method while the naive baseline has a smaller range of values compared with the other two. In several local regions (such as those marked by red boxes), the naive method even shows shapes different from those of the other two.}
    \label{fig:unsampling}
\end{figure}

\begin{table}
    \centering
    \caption{Performance comparison of three methods on ITER data, based on Python implementations. Magnetic-following interpolation results are taken as ground truth. The two $\phi$'s are selected because they show the best and worst approximation quality in the middle of two poloidal planes.
    }
    \small
    \begin{tabular}{c|c|c|c|c}
        \toprule
        $\phi$ & Interp & MSE & PSNR (dB) & Time (sec) \\ \midrule
        \multirow{3}*{$2\pi\times\frac{5}{32}$} & Magnetic following & - & - & 183,704 \\ \cline{2-5}
        & Ours & 2.714e-5 & 48.569 & 13.280 \\ \cline{2-5}
        & Naive & 2.632e-4 & 38.704 & 10.398 \\ \hline
        \multirow{3}*{$2\pi\times\frac{27}{32}$} & Magnetic following & - & - & 183,703 \\ \cline{2-5}
        & Ours & 2.544e-5 & 38.304 & 13.283 \\\cline{2-5}
        & Naive & 2.601e-4 & 28.208 & 10.398 \\ \bottomrule
    \end{tabular}
    \label{tab:performance}
\end{table}

We demonstrate and evaluate the use of our methodology with ITER data. Quantitative evaluation of W7-X is omitted here because they have different meshes on poloidal planes, which makes quantitative evaluation quite challenging.

\begin{figure}[!bh]
    \centering
    \includegraphics[width=\columnwidth]{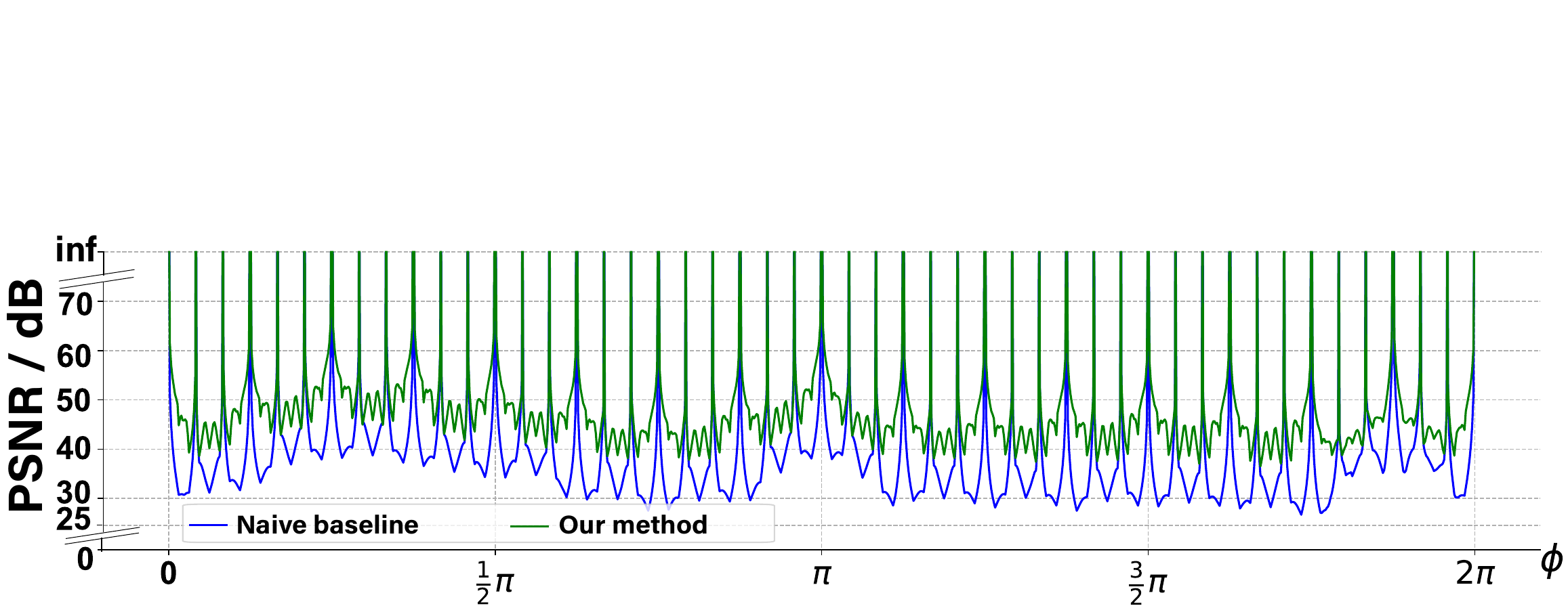}
    \caption{$\phi$ vs. PSNR for our method and naive baseline on ITER data taking magnetic-following interpolation as ground truth. Both curves have 48 infinity values because we upsampled the toroidal coordinate from 16 into 48 poloidal planes to avoid ill-posed prisms, and all three methods do the same interpolation on poloidal planes. The local minima are always attained at the middle of two poloidal planes for both our method and naive baseline, while the PSNR of our method is usually greater than 40~{dB} while that of the naive baseline drops below 30 dB sometimes.}
    \label{fig:phi_vs_psnr}
\end{figure}

\begin{figure}
    \centering
    \includegraphics[width=\columnwidth]{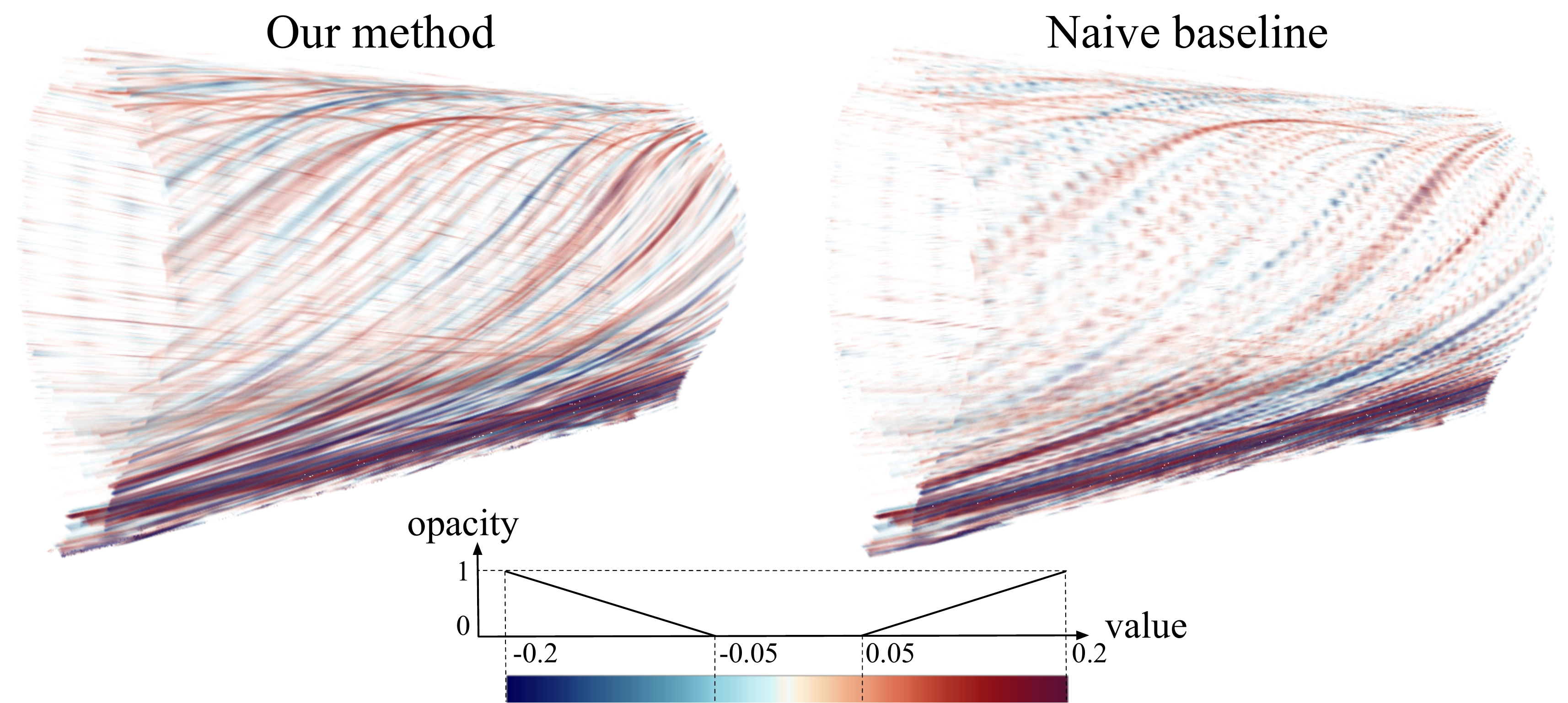}
    \caption{Volume rendering result of $f$. Transfer function is shown on the bottom. Our shows continuous trajectories of value while the naive baseline breaks the features into dashed lines, because it builds meaningless connectivities between two consecutive steps.}
    \label{fig:vol_ren}
\end{figure}

\begin{figure}
    \centering
    \includegraphics[width=\columnwidth]{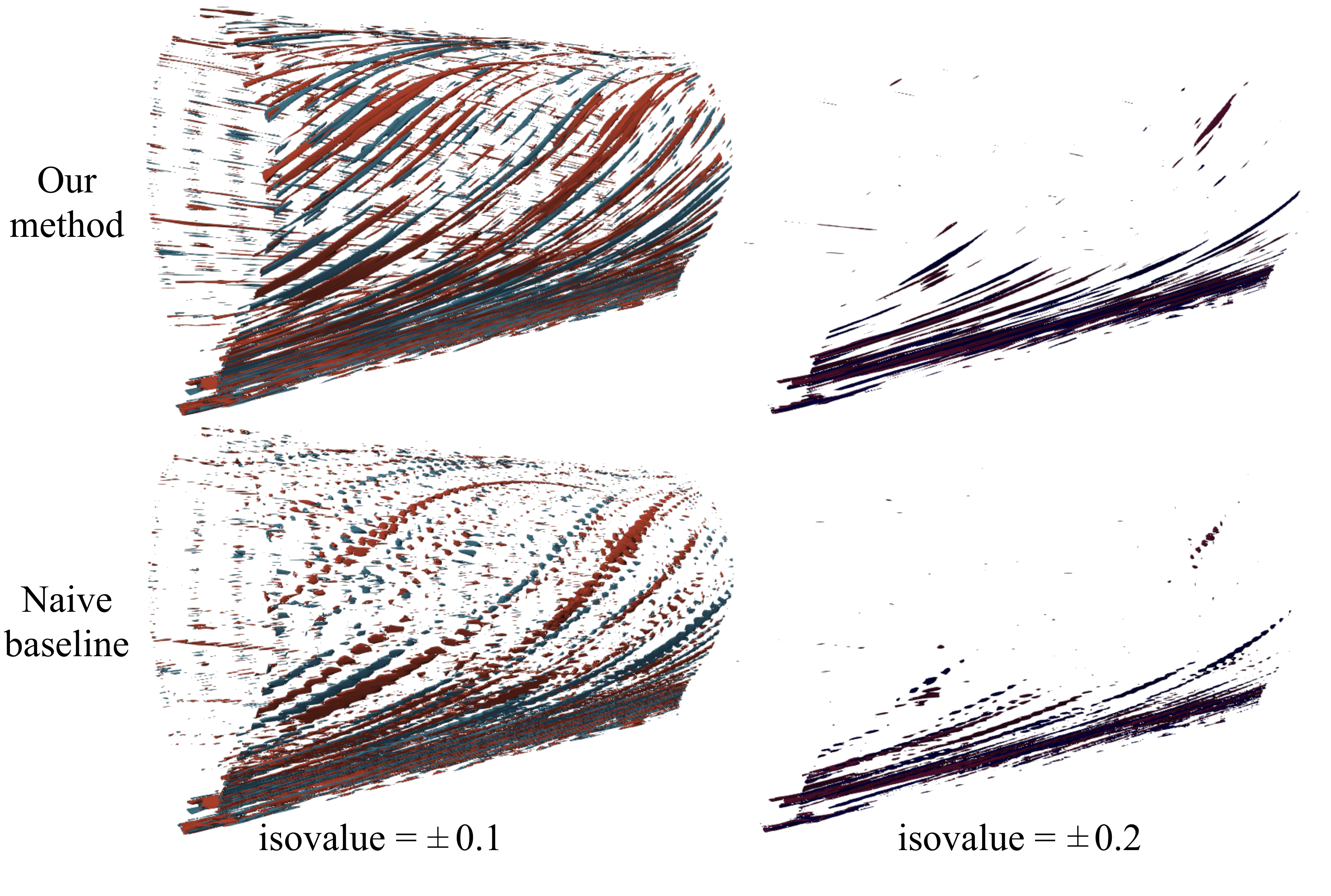}
    \caption{Isosurfaces with isovalues $\pm 0.1$ and $\pm 0.2$. Red lines show isosurfaces with positive isovalues, while blue lines show isosurfaces with negative isovalues. Well-chosen values can show more interesting blobs, which is beyond the scope of this work.
    }
    \label{fig:isosurfacing}
\end{figure}

\textbf{Evaluation scheme}. We define and use two types of approximations: \textit{forward approximation}, i.e., interpolation of function values $f$ at given 3D locations, such as volume rendering and particle tracing, which require sampling at arbitrary locations for ray/path integration, and \textit{inverse approximation}, i.e., root-finding of 3D locations with given function values, including isosurfacing in scalar-valued functions and critical-point-finding in vector-valued functions. We compare three types of interpolation schemes: (i) magnetic-following interpolation defined in Eq.~\eqref{eq:ff} that requires magnetic-line tracing, (ii) piecewise linear interpolation based on our simplicial meshing scheme that models deforming space, and (iii) a naive baseline based on simplicial meshing on uniform space where every node's next node is itself. The three interpolation schemes are measured by mean squared error (MSE), peak signal-to-noise ratio (PSNR), and running time. Compared with magnetic-following interpolation, our method significantly reduces the computational cost of all visualization tasks, because no magnetic-line tracing is involved for the interpolation.  We recognize that the piecewise linear (PL) approximation introduces error, and thus we treat the magnetic-following interpolation as the ground truth and quantitatively evaluate the error resulting from our method.  We also work together with domain scientists to evaluate our results qualitatively.

\textbf{Toroidal Upsampling (Forward Approximation)}. The objective of the toroidal upsampling case study is to upsample ITER's 3D scalar field data, which are simulated with relatively low resolution (e.g., $n_\phi=16$) uniformly along the toroidal direction at $\phi=2\pi i/16$ ($i=0,1,...,15$).  With each given $\phi$, science variables (e.g., density and temperature) are given at nodes on 2D triangular meshes.  We interpolate simulation outputs to obtain values of nodes at two arbitrary $\phi$'s, $(5/32)2\pi$ and $(27/32)2\pi$, which are both in the middle of two consecutive poloidal planes, to retrieve a much higher toroidal resolution of the data by the three interpolation methods. Results of the three methods and quantitative evaluation are shown in~\Cref{fig:unsampling} and~\Cref{tab:performance}, respectively. Our method shows a similar pattern of values, low MSE and high PSNR, and significantly less running time ($O(10^4)$ faster) when compared with magnetic-following interpolation. See how the running time varies over $\phi$ in Supplemental Materials.

We also uniformly took 1,024 poloidal planes over $\phi\in[0,2\pi)$ to show the tendency of PSNR with varying $\phi$ (\Cref{fig:phi_vs_psnr}), taking magnetic-following interpolation results as ground truth. Cells in the naive baseline do not conform with the variation in deforming space and thus interpolate the value of a position with irrelevant mesh nodes. Therefore, our method is expected to perform better than the baseline. As shown in~\Cref{fig:phi_vs_psnr}, The results of our method are close to ground truth when $\phi$ approaches poloidal planes and show worst accuracy when $\phi$ is in the middle of two poloidal planes. Also, the local minima of our method are greater than those of the naive baseline.

\textbf{Volume Rendering (Forward Approximation)}. We implement and compare two ways to volume render the scalar function $f$ in ITER data with the toroidal coordinate straightened, which transfers $(R,Z,\phi)^\intercal$ into $(x,z,y)^\intercal$, respectively, and renders in Cartesian coordinate system (\Cref{fig:vol_ren}). We can see a rotating pattern of values in both rendering images. However, the naive baseline fails to show a continuous pattern. The reason is that our method and naive baseline use different nodes in interpolation at the same location. Our method is based on a simplicial mesh of deforming space that relates each mesh node to those close to its location at the previous or next poloidal plane, while the naive baseline always uses nodes at the same location for interpolation, which fails to incorporate any deforming property in the underlying physics.

\textbf{Isosurfacing (Inverse Approximation)}. As shown in~\Cref{fig:isosurfacing}, isosurfaces show the same rotating pattern as volume rendering results in~\Cref{fig:vol_ren}. Again, our method gives continuous isosurfaces while the naive baseline yields ``dashed'' lines. Also, when we increase the isovalue to 0.2, then isosurfaces are mainly concentrated in only a small area that is around the saddle point in the magnetic field in ITER data, which also verifies the correctness of isosurfacing results in~\Cref{fig:isosurfacing}. In order to show more interesting features such as blobs, well-chosen isovalues are required, which is beyond the scope of this work.

\section{Discussion}
\label{sec:discussion}

Simplicial mesh for deforming 3D space makes it possible to derive function values at any given 2D position and any given $\phi$ simply by PL interpolation. It gives results with comparable quality (measured by MSE and PSNR) compared with those given by physical derivation while taking significantly less time. It also shows why the naive interpolation is not applicable: our method builds reasonable connectivities between two successive poloidal planes, whereas the naive interpolation relates a 3D point with irrelevant mesh nodes. Furthermore, it enables the inverse approximation, a solution for root-finding problems, which cannot be achieved by tracing magnetic lines.

\textbf{Difference between our method and general-purpose tetrahedralization methods}. Our method and general-purpose tetrahedralization methods differ in goals, constraints, algorithms, and time complexity. First, we aim for a fast and invertible approximation of deforming 3D space, constrained by interior node connectivities that represent the physical behavior of deformation. Differently, general-purpose tools like TetGen and Gmsh mainly focus on filling a 3D domain bounded by a surface triangular mesh, often encountering difficulties with twisting interior edge constraints. Second, our method follows a divide-and-conquer paradigm that uses interior edges to partition the domain, resulting in lower time complexity ($O(|\mathcal{V}|)$) as well as the potential for parallelization. However, general-purpose methods, typically based on Delaunay triangulation, exhibit higher time complexity (up to $O(|\mathcal{V}|^2)$ in our problem) and are less amenable to parallelization. Finally, general-purpose methods add Steiner points to improve tetrahedral quality (i.e., how close a tetrahedron is to an equilateral shape), but these Steiner points are unnecessary in our problem, as connectivity constraints alone guarantee sufficient approximation accuracy.

\textbf{Limitations}. Although this work does not require the identity of spatial meshes over $\phi$ or any properties regarding the next node connectivities, there are still several limitations.  First, the decision tree search can be improved, such as by extending the search of new connections to $n$-hop neighbors, finding a better heuristic for choosing a pivot node, and parallelizing the branch search to reduce its time complexity for the worst case. Second, the indivisible problem is not completely solved. For example, we do not triangulate ill-posed prisms but instead increase the toroidal resolution to avoid them. Also, this work provides no methods to check indivisible polyhedra except for prisms without dividing them first. Third, we do not continue with figuring out the optimal balance between the number of nodes in cycles formed by cutting edges and the complexity of triangulating and coordinating the cross sections.

\textbf{Future work}. Besides solving the previously mentioned limitations, this work has several possible extensions. One possible way is to extend the problem from deforming 3D space to multi-D deforming space. For example, if we have a series of 3D tetrahedral meshes as input, then it extrudes a 4D space. Triangulating a 4-polytope into 4-simplices with all similar constraints involves more problems. Another promising work could be extending the 2D triangular mesh to 2D mesh with any polygon as cells. Additional work is triangulating the polygonal cells while guaranteeing that the whole space is still triangulable. A third future work could be using other interpolation methods, such as neural networks, to increase interpolation quality further. PL interpolation may not be the best model of value distribution in cells. Although we can increase toroidal resolution to improve interpolation accuracy, upsampling via magnetic-following interpolation is still time-consuming.

\section{Conclusion}
\label{sec:conclusion}

We propose an algorithm based on the divide-and-conquer paradigm to triangulate 3D nonconvex deforming space with geometric and connectivity constraints. Specifically, in the divided stage, we split the space into smaller 3D partitions to reduce time complexity. Cutting edges are well chosen to guarantee the independence of subdividing any partition; communication between partitions on cross sections is also taken care of. In the conquer stage, we eliminate nodes in a 3D partition one by one through a decision-tree-based search. This tree makes decisions on which node to eliminate and how to eliminate it. Our algorithm provides a scheme for and reduces the complexity of many visualization and analysis tasks, such as upsampling, volume rendering, and isosurfacing, on deforming 3D space, especially those with rotating features. We evaluate the algorithm both quantitatively and qualitatively on XGC datasets, which verifies its high accuracy compared with naive space interpolation and low time complexity compared with TetGen and the magnetic-line-following method.

\acknowledgments{%
he authors thank Drs. Choong-Seock Chang, Michael Churchill, Robert Hager, Seung-Hoe Ku, Zeyu Guo, and Rephael Wenger for insightful discussions. This research is supported by DOE DE-SC0022753.  It is also supported by the U.S. Department of Energy, Office of Advanced Scientific Computing Research, Scientific Discovery through Advanced Computing (SciDAC) program of the U.S. Department of Energy under Contract No. DE-AC02-06CH11357.
}

\bibliographystyle{abbrv-doi-hyperref}

\bibliography{references}



\end{document}


\firstsection{Additional Evaluation}
\maketitle

\begin{figure}[!th]
    \centering
    \includegraphics[width=\columnwidth]{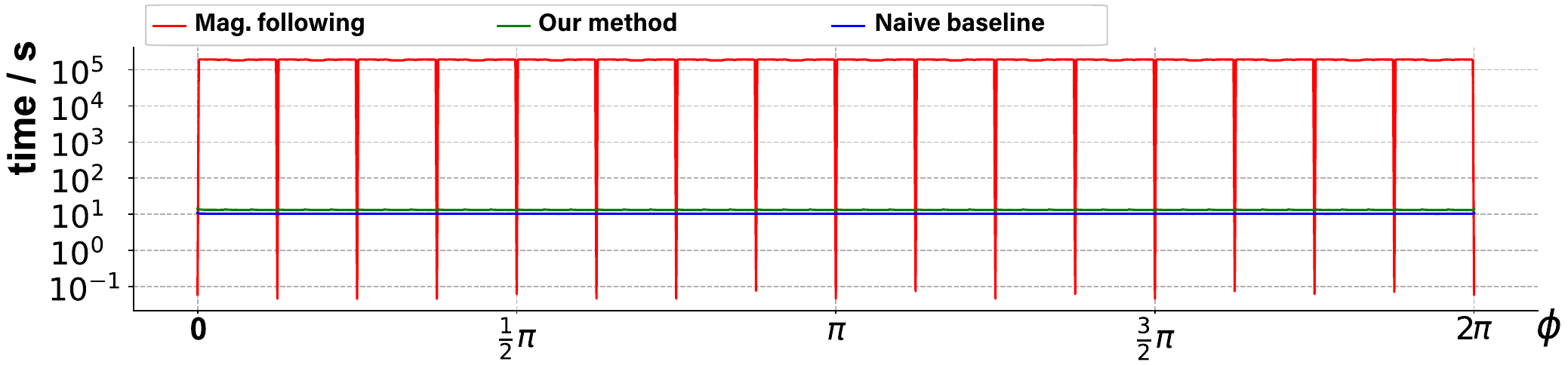}
    \caption{
    $\phi$ vs. time for magnetic-following interpolation, our method, and the naive baseline on XGC data. Naive baseline is slightly faster than our method, while magnetic-following method usually needs $O(10^4)$ times compared with the other two. On the sampled poloidal planes, magnetic-following method takes less than one second because only a 2D cell locating and a 2D barycentric interpolation are required for it on sampled poloidal planes.
    }
    \label{fig:phi_vs_time}
\end{figure}

\textbf{Extremum Lines (Inverse Approximation)}. To study 3D blob filaments, we co-designed with scientists the definition of \emph{blob core lines} as the extremum lines---loci of local minimum/maximum where radial and axial gradients of $f$ vanish and the radial-axial Hessian $\mathbf{H}_f$ is positive-/negative-definite:
\begin{equation}
  \frac{\partial f}{\partial R} = \frac{\partial f}{\partial Z} = 0 \mbox{~and~} \lambda_1\lambda_2>0,
\end{equation}
where $R$ and $Z$ are the radial and axial coordinates, respectively, and 
$\lambda_1$ and $\lambda_2$ are the eigenvalues of the radial-axial Hessian $\mathbf{H}_f$, which considers only $R$ and $Z$ axes:
\begin{equation}
  \mathbf{H}_f = 
  {\begin{pmatrix} \frac{\partial^2 f}{\partial R^2} & \frac{\partial^2 f}{\partial R\partial Z}\\
  \frac{\partial^2 f}{\partial Z\partial R} & \frac{\partial^2 f}{\partial Z^2}
  \end{pmatrix}}.
\end{equation}
%
We reformulate the extraction of extremum lines as a critical point tracking problem by treating $\phi$ as time and use the feature tracking kit FTK~\cite{guo2021ftk} to extract and visualize the curves.  FTK assumes that the gradient vector field is piecewise linear so that critical point trajectories are reconstructed directly in a space mesh, which is by default nondeforming; we made minor changes to FTK software to support deformed space meshes.  Figure~\ref{fig:filament} demonstrates extremum line extraction results with our deformed mesh.  Qualitatively speaking, the resulting curves reflect the same trends that scientists observe in isosurface and volume rendering; we plan to investigate further the accuracy evaluation of extremum lines in future work.




\begin{figure}
    \centering
    \includegraphics[width=0.65\columnwidth]{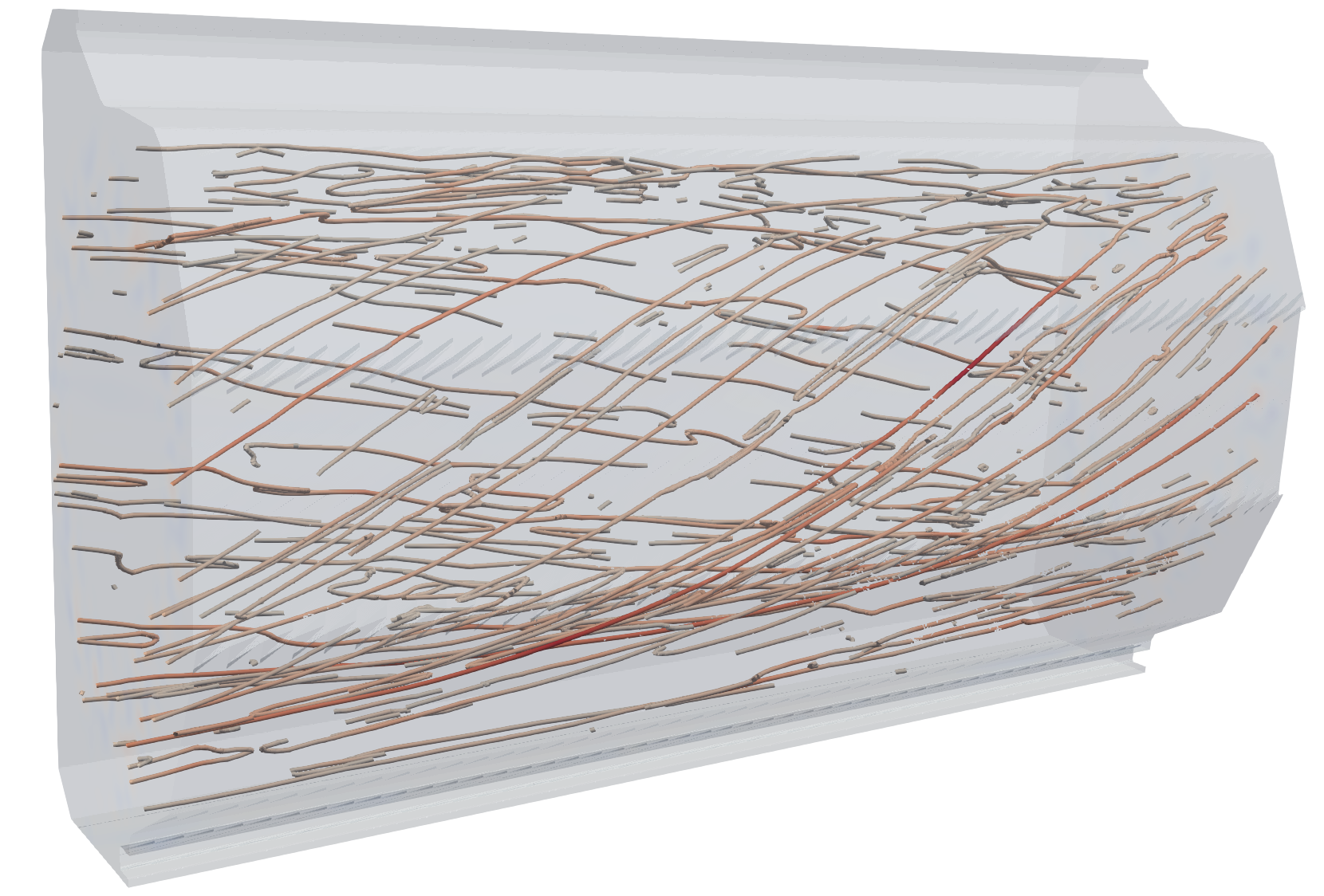}
    \caption{Extremum line extraction in XGC data with our deformed mesh.} 
    \label{fig:filament}
\end{figure}

\section{Additional Case Study with Synthetic Data}

The additional case study demonstrates the generality of our methodology beyond fusion. The data are synthesized by rotating two Gaussian blobs with phase difference $\pi$ around the center of a circular field. We set the radius of the circular field to be 1 and the full width at half maximum of two blobs to be 0.3, and we uniformly choose 16 $\phi$'s in one cycle ($2\pi$). 
The synthetic function is 
\begin{align}
\begin{split}
    f(x,y,t)&=Ce^{-B((x-A\cos\omega t)^2+(y-A\sin\omega t)^2))}\\
    &+Ce^{-B((x+A\cos\omega t)^2+(y+A\sin\omega t)^2))},
\end{split}
\end{align}
%
where $C=9.806$, $B=30.807$, and $\omega=0.393$.  
Values at several timesteps are shown in \Cref{fig:synthetic}. Discretization of this 2D field gives 1,000 nodes and 1,903 triangle cells (\Cref{fig:synthetic_mesh}~{(a)}). Next nodes of nodes are determined by their positions at the next timestep, which forms a space with 2,000 nodes (\Cref{fig:synthetic_mesh}~{(c)}). Our algorithm yields 5,961 tetrahedron cells for this space. We conduct the same evaluation as in the main document on this synthetic dataset.

\begin{figure}
    \centering
    \includegraphics[width=\columnwidth]{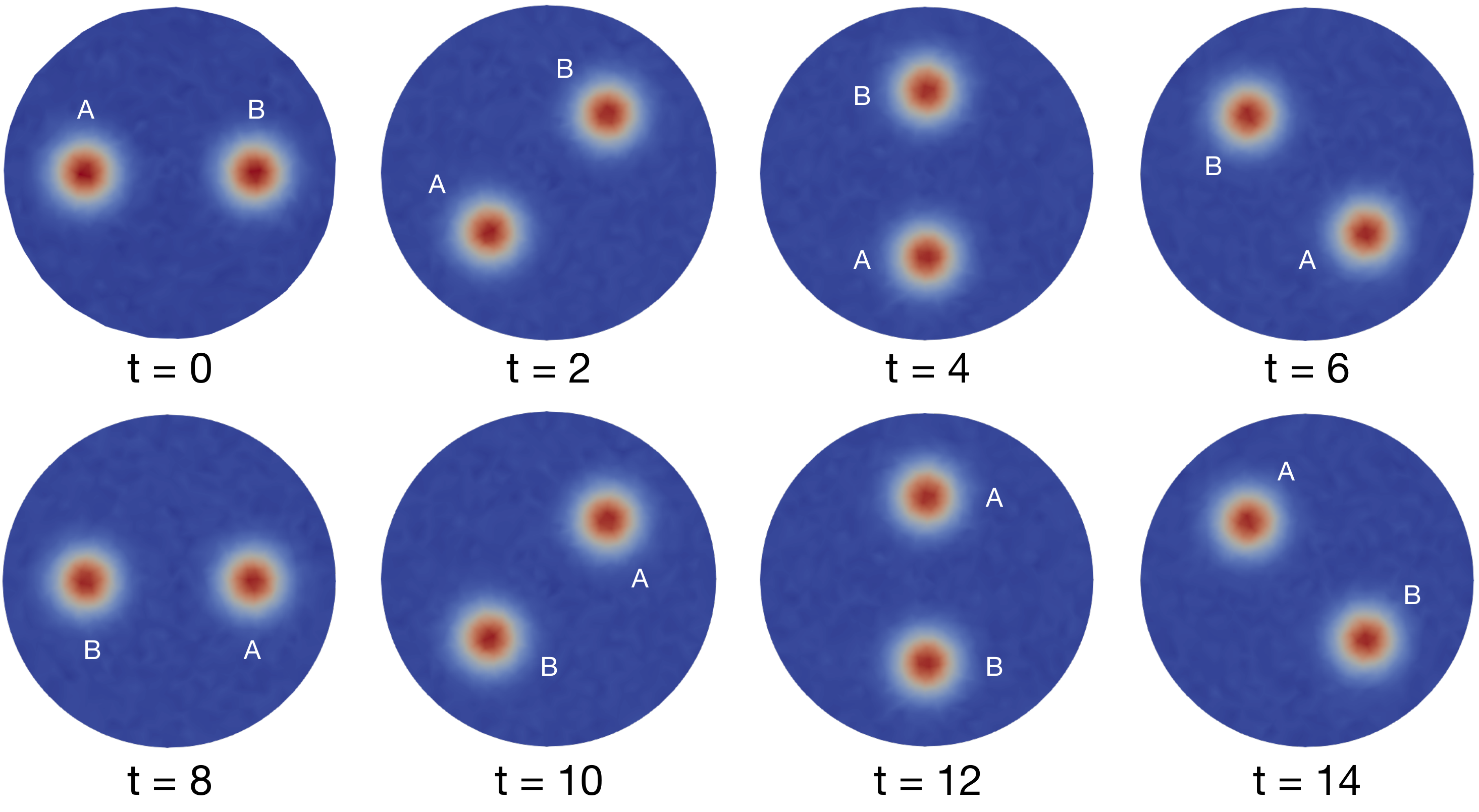}
    \caption{Synthetic dataset showing rotation of two Gaussian blobs, denoted by A and B, around the center of a circular field with cycle $T=16$. This time-varying dataset also introduces a deforming space.}
    \label{fig:synthetic}
\end{figure}

\begin{figure}
    \centering
    \includegraphics[width=\columnwidth]{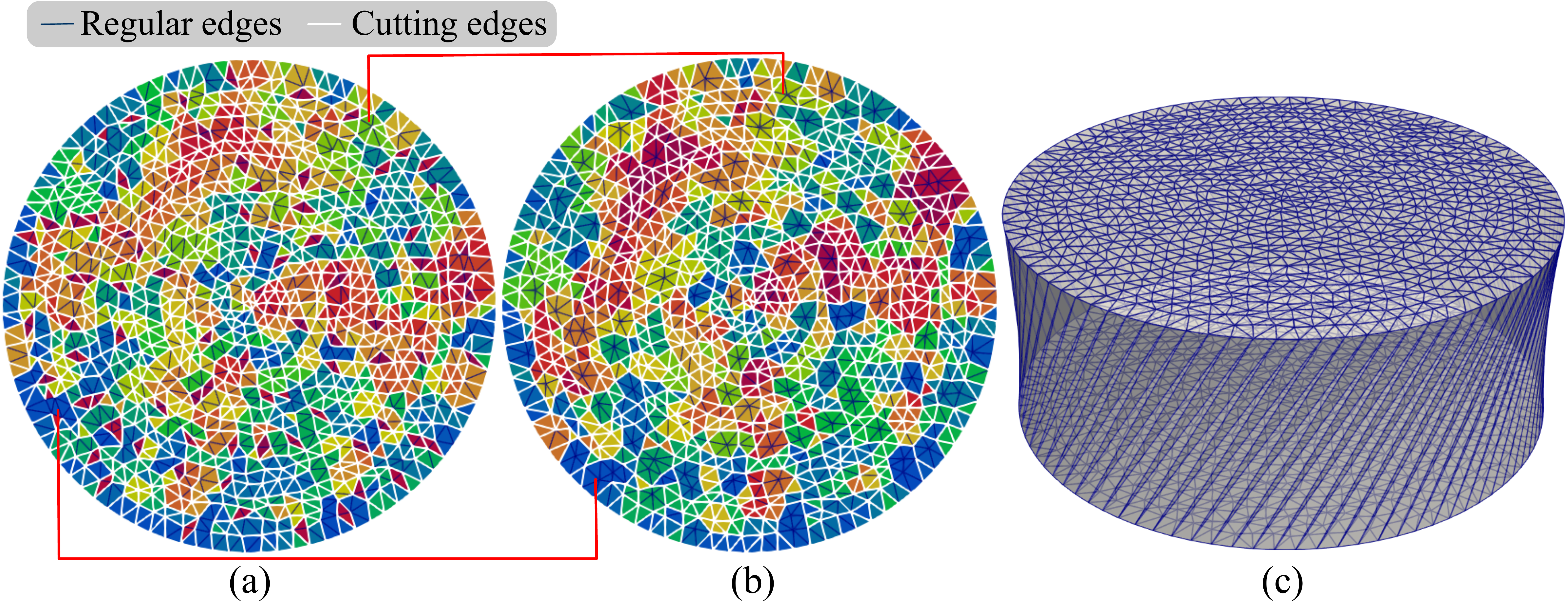}
    \caption{
    Mesh for synthetic data. (a) and (b) show the original 2D triangular mesh of a circular field with 1,000 nodes and 1,903 cells. Spatial-connected components in lower and upper meshes are differentiated by a categorical colormap in (a) and (b), respectively, while white edges highlight cutting graphs. Every red line connects two matched components. (c) shows the deformed space mesh with 2,000 nodes. Lateral edges indicate the rotation of nodes. All edges are straight lines, although they create a curved visual effect on two sides. The deformed mesh generated by our algorithm yields 5,841 cells for these 2,000 nodes.}
    \label{fig:synthetic_mesh}
\end{figure}

\begin{figure}
    \centering
    \includegraphics[width=\columnwidth]{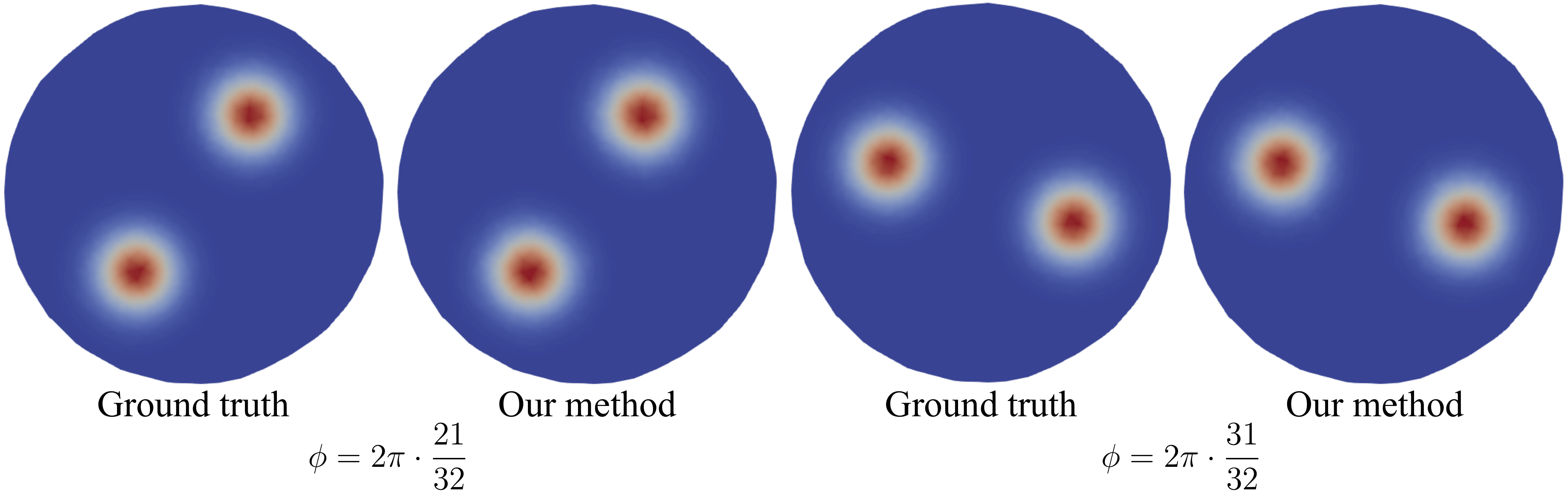}
    \caption{Ground truth and upsampling results of our method on the synthetic dataset. Our method at $\phi=(21/32)2\pi$ has MSE 0.074 and PSNR 50.718 dB; that at $\phi=(31/32)2\pi$ has MSE 0.191 and PSNR 50.593 dB.}
    \label{fig:synthetic_unsampling}
\end{figure}

\textbf{Temporal upsampling}. We again choose two arbitrary timesteps that are in the middle of two successive sampled timesteps, corresponding to $\phi=(21/32)2\pi$ and $\phi=(31/32)2\pi$ in toroidal coordinates. The comparison of our method and ground truth is shown in \Cref{fig:synthetic_unsampling}, which shows a high PSNR. \Cref{fig:synthetic_phi_vs_psnr} shows the variation of PSNR over toroidal coordinates on both our method and the naive baseline. Since there are 16 sampled timesteps, there are 16 arcs in \Cref{fig:synthetic_phi_vs_psnr}. Each arc attains its maxima at two sampled toroidal angles and its minimum at the middle of these two angles. Also, our method has a higher minimum PSNR than does the naive baseline, meaning the naive baseline gives wrong interpolation results.

\begin{figure}[!th]
    \centering
    \includegraphics[width=\columnwidth]{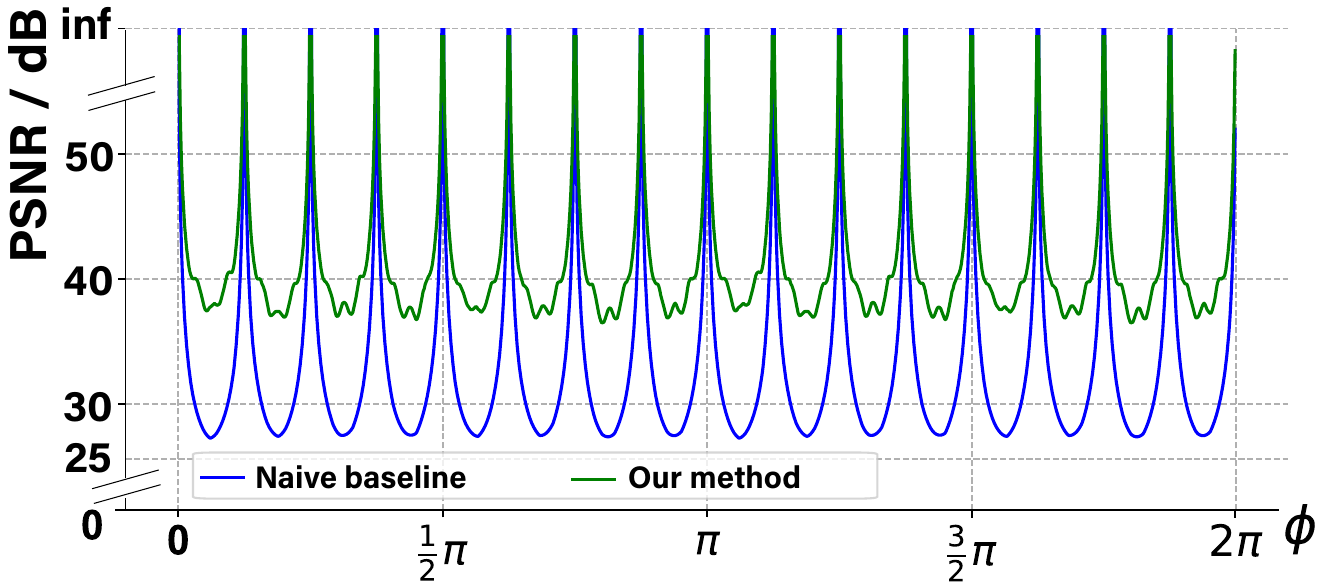}
    \caption{$\phi$ vs. PSNR of our method and the naive baseline on synthetic dataset. We evenly sample 16 timesteps, so there are 16 infinity values on both curves. All other observations agree with those in the main document.}
    \label{fig:synthetic_phi_vs_psnr}
\end{figure}

\begin{figure}[!th]
    \centering
    \includegraphics[width=0.75\columnwidth]{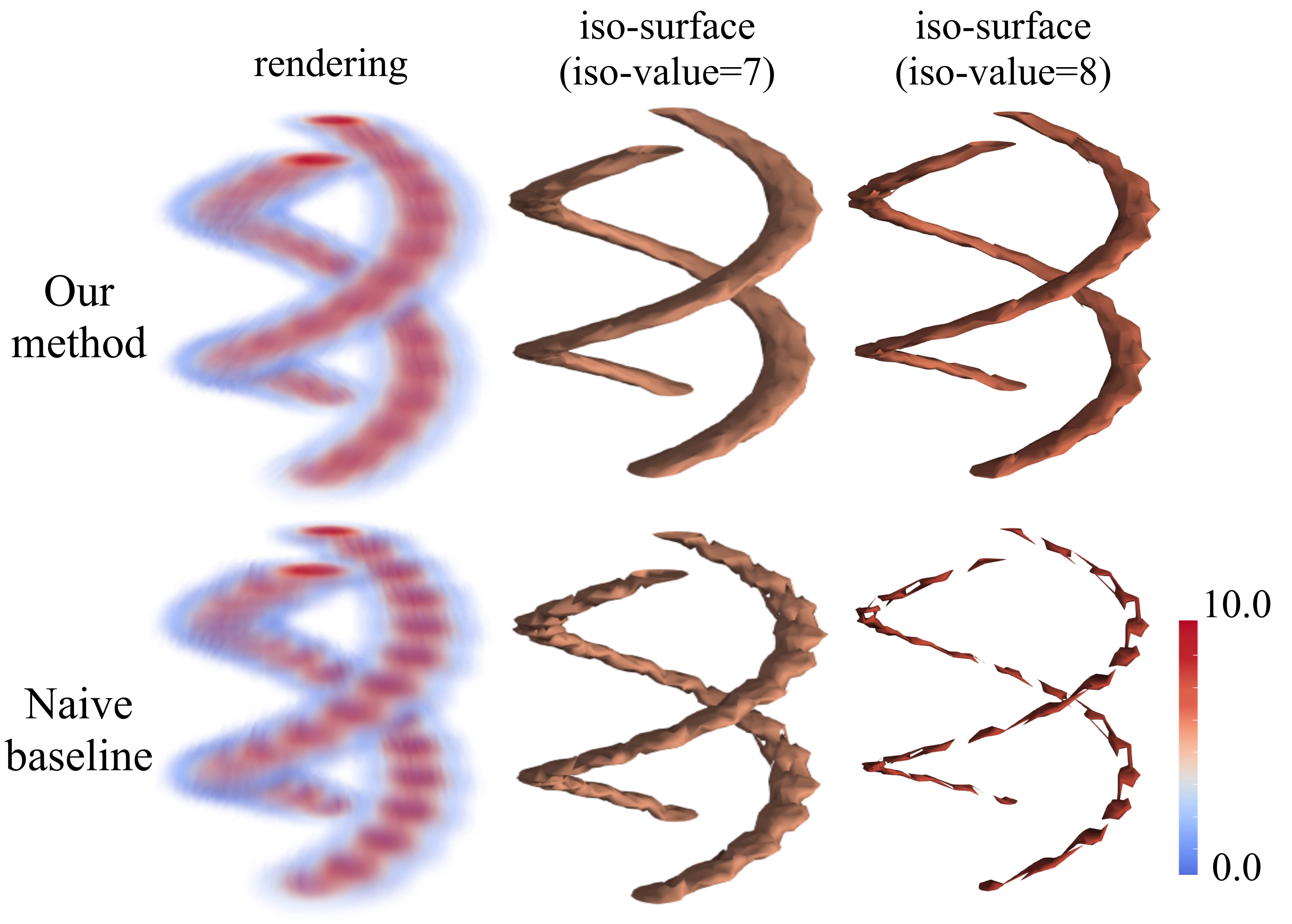}
    \caption{Volume rendering results and isosurfaces for our method and the naive baseline on synthetic data. Similar to the results in the main document, our method yields continuous features while the naive baseline breaks the features.}
    \label{fig:synthetic_vol}
\end{figure}

\textbf{Volume rendering and isosurfacing}. Results for both our method and naive interpolation are shown in \Cref{fig:synthetic_vol}. 
our method reflects the correct rotation of two blobs, but the naive baseline gives a broken translation along the time axis.

\newpage

\section{Complete version of lists and figures}

\begin{table}[!t]
  \caption{%
  	Types of neighbors around pivot nodes, their triangulation ways, and corresponding separated tetrahedra. Nodes with superscripts are on the upper mesh while those without superscripts are on the lower mesh. Red lines represent newly added edges in triangulation. ``Pviot'' in the last column is replaced by ``p'' for short.
  }
  \label{tab:pivot_cases}
  \scriptsize%
  \centering%
  \begin{tabular}{m{0.1in}|m{0.2in}|m{0.8in}|m{0.85in}|m{0.5in}}
  	\toprule
  	Min deg & Pivot node & \tabincell{pivot node and\\its neighbors} & Volume triangulations & Separated tetrahedra \\
   \midrule
  	\multirow{2}*{3} & \tabincell{in\\lower\\mesh} & \includegraphics[width=0.7in]{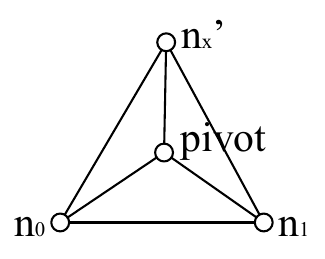} & - & \{$p$, $n_0$, $n_1$, $n_x'$\} \\ \cline{2-5}
  	& \tabincell{in\\upper\\mesh} & \includegraphics[width=0.7in]{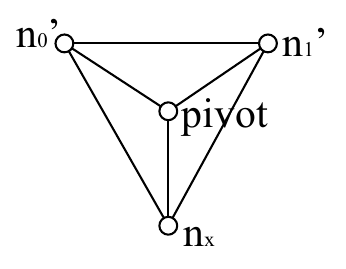} & - & \{$p$, $n_0'$, $n_1'$, $n_x$\} \\ \hline
  	\multirow{10}*{4} & \multirow{5}*{\tabincell{in\\lower\\mesh}} & \multirow{2}*{\includegraphics[width=0.7in]{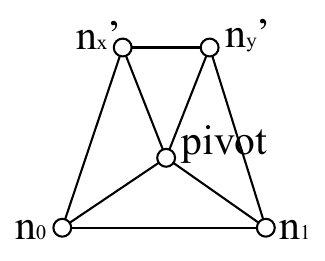}} & \includegraphics[width=0.7in]{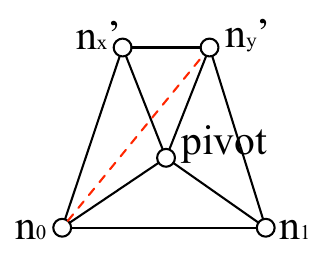} & \tabincell{\{$p$, $n_0$,\\$n_x'$, $n_y'$\},\\\{$p$, $n_0$,\\$n_1$, $n_y'$\}}\\ \cline{4-5}
  	& & & \includegraphics[width=0.7in]{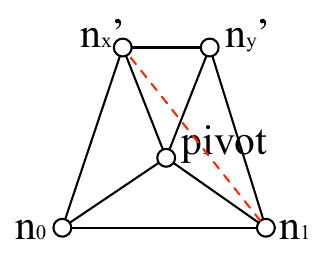} & \tabincell{\{$p$, $n_0$,\\$n_1$, $n_x'$\},\\\{$p$, $n_1$,\\$n_x'$, $n_y'$\}} \\ \cline{3-5}
  	& & \includegraphics[width=0.7in]{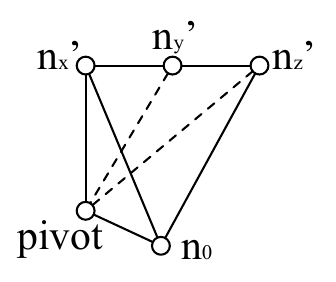} & \includegraphics[width=0.7in]{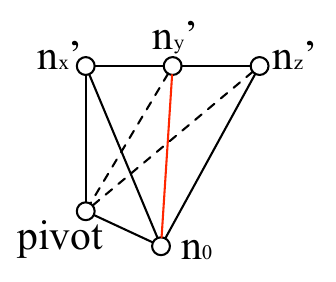} & \tabincell{\{$p$, $n_0$,\\$n_x'$, $n_y'$\},\\\{$p$, $n_0$,\\$n_y'$, $n_z'$\}} \\ \cline{3-5}
  	& & \includegraphics[width=0.6in]{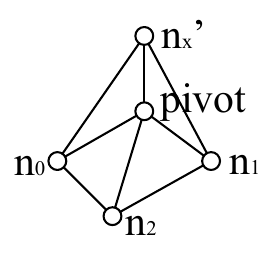} & \includegraphics[width=0.7in]{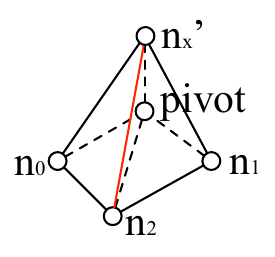} & \tabincell{\{$p$, $n_0$,\\$n_2$, $n_x'$\},\\\{$p$, $n_1$,\\$n_2$, $n_x'$\}} \\ \cline{3-5}
  	& & \includegraphics[width=0.6in]{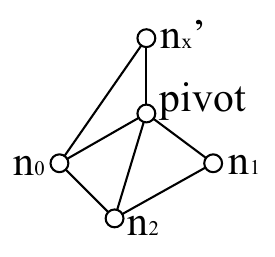} & \includegraphics[width=0.7in]{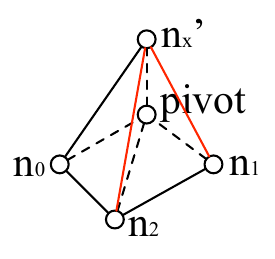} & \tabincell{\{$p$, $n_0$,\\$n_2$, $n_x'$\},\\\{$p$, $n_1$,\\$n_2$, $n_x'$\}} \\ \cline{2-5}
  	& \multirow{5}*{\tabincell{in\\upper\\mesh}} & \multirow{2}*{\includegraphics[width=0.7in]{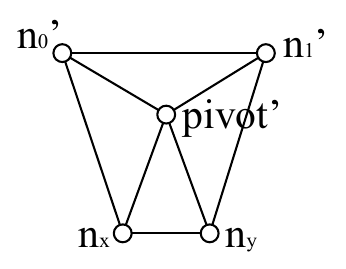}} & \includegraphics[width=0.7in]{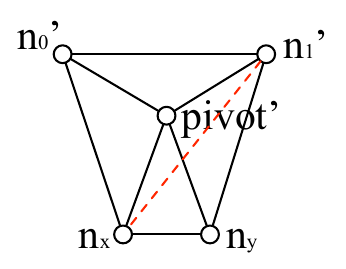} & \tabincell{\{$p$, $n_x$,\\$n_0'$, $n_1'$\},\\\{$p$, $n_x$,\\$n_y$, $n_1'$\}}  \\ \cline{4-5}
  	& & & \includegraphics[width=0.7in]{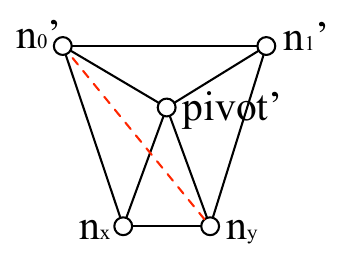} & \tabincell{\{$p$, $n_x$,\\$n_y$, $n_0'$\},\\\{$p$, $n_y$,\\$n_0'$, $n_1'$\}}  \\ \cline{3-5}
  	& & \includegraphics[width=0.6in]{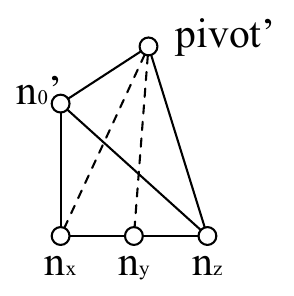} & \includegraphics[width=0.6in]{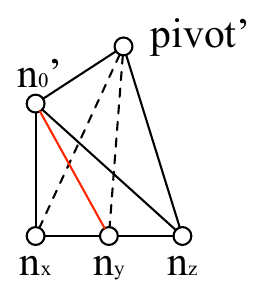} & \tabincell{\{$p$, $n_x$,\\$n_y$, $n_0'$\},\\\{$p$, $n_y$,\\$n_z$, $n_0'$\}} \\ \cline{3-5}
  	& & \includegraphics[width=0.6in]{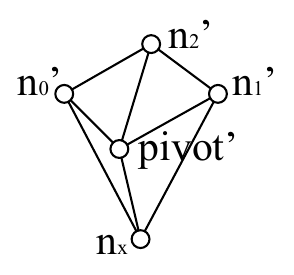} & \includegraphics[width=0.7in]{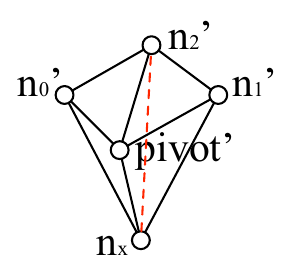} & \tabincell{\{$p$, $n_x$,\\$n_0'$, $n_2'$\},\\\{$p$, $n_x$,\\$n_1'$, $n_2'$\}} \\ \cline{3-5}
  	& & \includegraphics[width=0.6in]{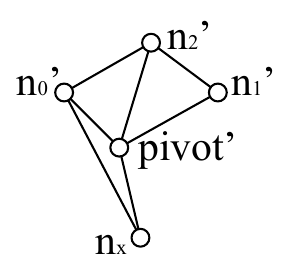} & \includegraphics[width=0.7in]{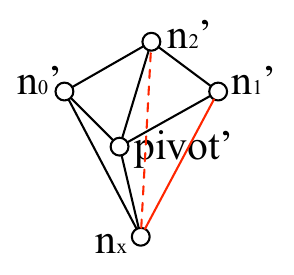} & \tabincell{\{$p$, $n_x$,\\$n_0'$, $n_2'$\},\\\{$p$, $n_x$,\\$n_1'$, $n_2'$\}} \\
  	\bottomrule
  \end{tabular}%
\end{table}

\begin{figure}[!ht]
    \centering
    \includegraphics[width=0.458\textwidth]{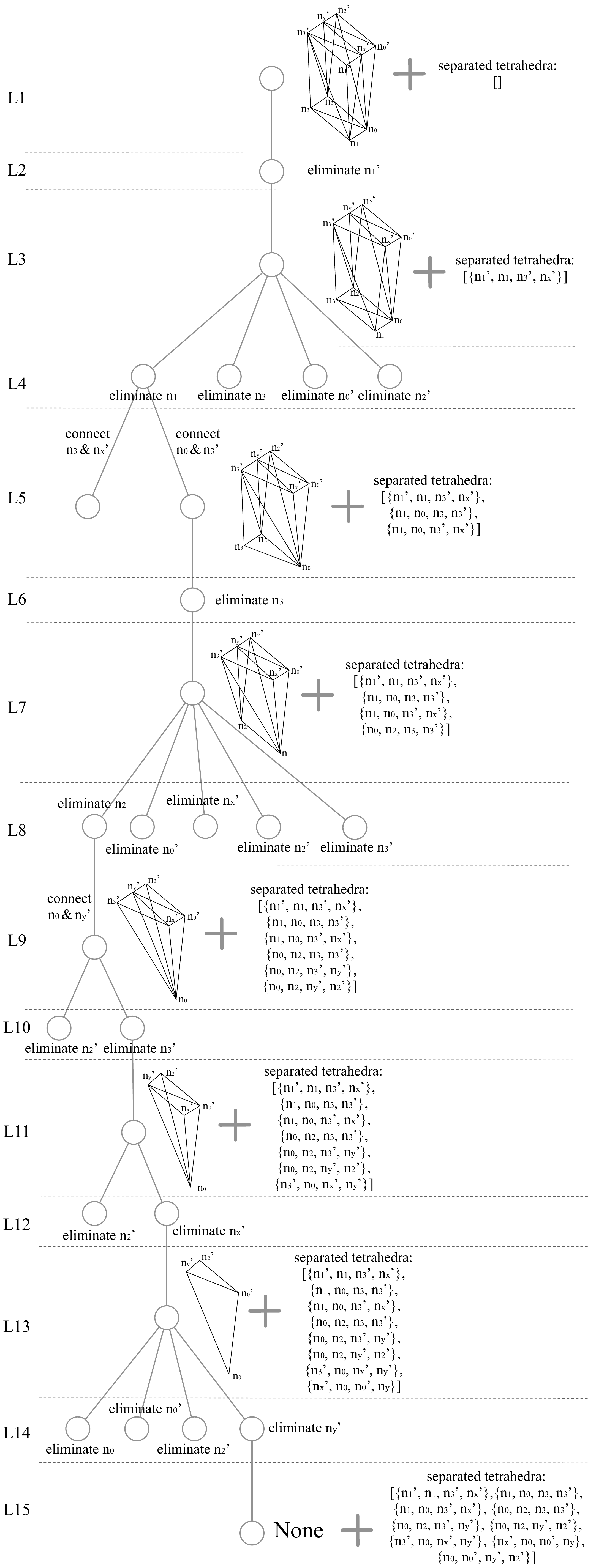}
    \caption{Node elimination algorithm shown in a decision-tree paradigm. In odd levels, it makes a decision on ``which mesh node to eliminate'' among mesh nodes with minimum degree; in even levels, it makes a decision on ``how to eliminate'' based on the isolated polyhedron formed by the pivot node and its neighbors.
    }
    \label{fig:node-elimination}
\end{figure}

\begin{table}[h]
  \caption{%
  	Check list for valid subdivision of deformed prisms. If a triangulation is valid, then the two corresponding conditions regarding whether two points are on different sides of a plane shown in the last three columns should  both be satisfied. If one of the points is in the plane, then there is a tetrahedron degenerating into a quadrilateral. The example shown in the third column is an ill-posed prism, which is created by two equilateral triangles where one can be aligned with the other by rotating it by $\pi/2$ around its centroid. The ill-posed prism cannot be triangulated by any of the six triangulations.
  }
  \label{tab:ill_cases}
  \centering%
  \begin{tabular}{m{0.1in}|m{0.45in}|m{0.8in}|m{0.2in}|m{0.2in}|m{0.6in}}
  	\toprule
  	\multirow{2}*{ID} & \multirow{2}*{\tabincell{Face\\separations}} & \multirow{2}*{Subdivision} & \multicolumn{3}{c}{Conditions to be checked} \\ \cline{4-6}
   & & & pt 1 & pt 2 & plane \\
   \midrule
   \end{tabular}
  \renewcommand{\arraystretch}{2.3}
  \begin{tabular}{m{0.1in}|m{0.45in}|m{0.8in}|m{0.2in}|m{0.2in}|m{0.6in}}
  	\multirow{3}*{1} & $a_0-b_1$ & \multirow{3}*{\includegraphics[width=0.8in]{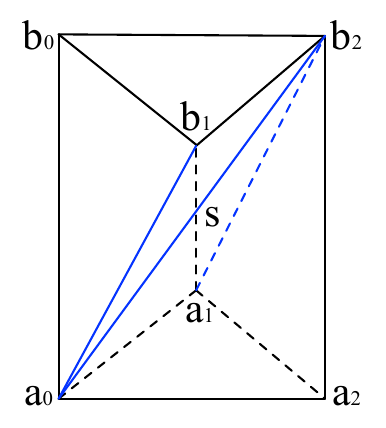}} & $a_2$ & $b_1$ & $a_0-a_1-b_2$ \\
   & $a_1-b_2$ & & & & \\
   & $a_0-b_2$ & & $a_1$ & $b_0$ & $a_0-b_1-b_2$ \\ \hline
  	\multirow{3}*{2} & $a_0-b_1$ & \multirow{3}*{\includegraphics[width=0.8in]{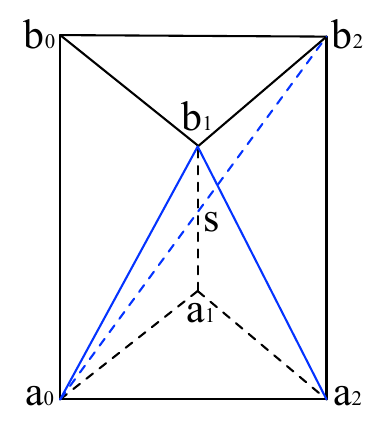}} & $a_1$ & $b_2$ & $a_0-a_2-b_1$ \\
   & $a_2-b_1$ & & & & \\
   & $a_0-b_2$ & & $a_2$ & $b_0$ & $a_0-b_1-b_2$ \\ \hline
  	\multirow{3}*{3} & $a_1-b_0$ & \multirow{3}*{\includegraphics[width=0.8in]{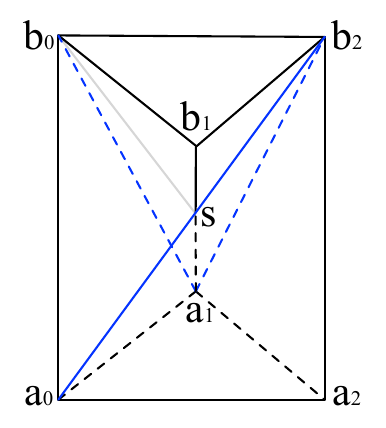}} & $a_2$ & $b_0$ & $a_0-a_1-b_2$ \\
   & $a_1-b_2$ & & & & \\
   & $a_0-b_2$ & & $a_0$ & $b_1$ & $a_1-b_0-b_2$ \\ \hline
  	\multirow{3}*{4} & $a_0-b_1$ & \multirow{3}*{\includegraphics[width=0.8in]{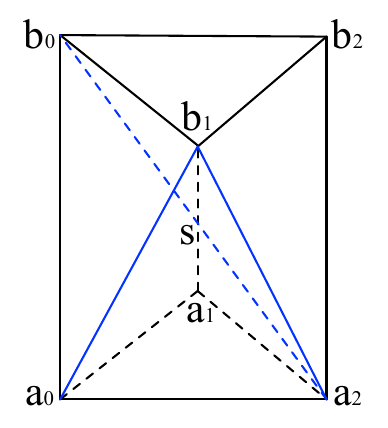}} & $a_1$ & $b_0$ & $a_0-a_2-b_1$ \\
   & $a_2-b_1$ & & & & \\
   & $a_2-b_0$ & & $a_0$ & $b_2$ & $a_2-b_0-b_1$ \\ \hline
  	\multirow{3}*{5} & $a_1-b_0$ & \multirow{3}*{\includegraphics[width=0.8in]{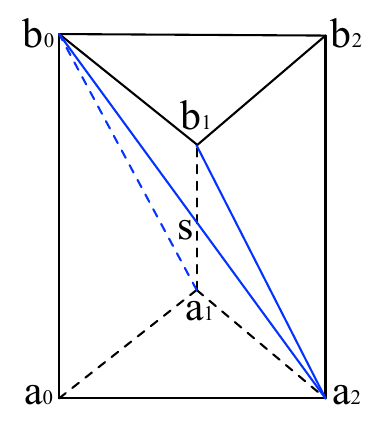}} & $a_0$ & $b_1$ & $a_1-a_2-b_0$ \\
   & $a_2-b_1$ & & & & \\
   & $a_2-b_0$ & & $a_1$ & $b_2$ & $a_2-b_0-b_1$ \\ \hline
  	\multirow{3}*{6} & $a_1-b_0$ & \multirow{3}*{\includegraphics[width=0.8in]{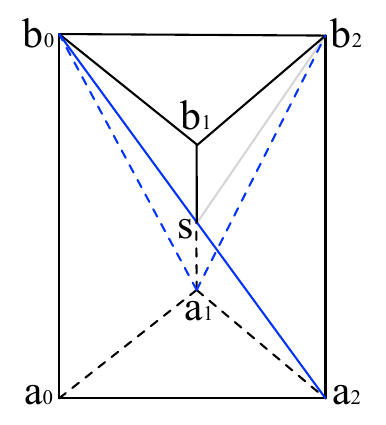}} & $a_0$ & $b_2$ & $a_1-a_2-b_0$ \\
   & $a_1-b_2$ & & & & \\
   & $a_2-b_0$ & & $a_2$ & $b_1$ & $a_1-b_0-b_2$ \\
  	\bottomrule
  \end{tabular}%
\end{table}

\bibliographystyle{abbrv-doi-hyperref}

\bibliography{references}

